\documentclass{LMCS}

\def\dOi{11(2:11)2015}
\lmcsheading%
{\dOi}
{1--39}
{}
{}
{May\phantom.~19, 2014}
{Jun.~19, 2015}
{}

\ACMCCS{[{\bf Theory of computation}]: Logic}

\usepackage{epsfig, amssymb}
\usepackage{amsmath}
\usepackage{latexsym, graphics, graphicx}
\usepackage{yfonts}
\usepackage{proof}
\usepackage{algorithmic}
\usepackage{algorithm}
\usepackage{amsthm} 
\usepackage{amsxtra}
\usepackage{color}
\usepackage{float}
\usepackage{hyperref}

\newcommand{\ignore}[1]{}

\newcommand{\eq}{\mathcal{EQ}}

\newcommand{\ueq}{=_{}^?}

\theoremstyle{plain}
\newtheorem{theorem}{Theorem}[section]
\newtheorem{lemma}[theorem]{Lemma}

\theoremstyle{definition}
\newtheorem{definition}[theorem]{Definition}

\newtheorem{example}[theorem]{Example}

\begin{document}

\title[Unification Modulo One-Sided Distributivity]{ On Unification
  Modulo One-Sided Distributivity: Algorithms, Variants and
  Asymmetry\rsuper*}

\author[A.~M.~Marshall]{Andrew M. Marshall\rsuper a}	
\address{{\lsuper a}University of Mary Washington}
\email{marshall@umw.edu}  

\author[C.~Meadows]{Catherine Meadows\rsuper b}	
\address{{\lsuper b}Naval Research Laboratory}	
\email{catherine.meadows@nrl.navy.mil}  

\author[p.~Narendran]{Paliath Narendran\rsuper c}	
\address{{\lsuper c}University at Albany--SUNY}	
\email{dran@cs.albany.edu}  


\titlecomment{{\lsuper*}A preliminary version of a portion of this work appeared as~\cite{marshall12c}}

\begin{abstract}
An algorithm for unification modulo {\em one-sided\/} distributivity 
is an early result by Tid\'{e}n and Arnborg.  
More recently this theory has been of interest in cryptographic protocol analysis 
due to the fact that many cryptographic operators satisfy this property.  
Unfortunately the algorithm presented in the paper, although correct, 
has recently been shown not to be polynomial time bounded as 
claimed. 
In addition, for some instances,  
there exist most general unifiers that are exponentially large 
with respect to the input size. In this paper we 
first present a new polynomial 
time algorithm that solves the \emph{decision problem} for
a non-trivial subcase, based on a typed theory, of 
unification modulo one-sided distributivity.
Next we present a new polynomial algorithm that
solves the decision problem for unification modulo one-sided distributivity. 
A construction, employing string compression, is used to achieve the
polynomial bound. Lastly, we examine the one-sided distributivity problem in the 
new asymmetric unification paradigm. We give the first asymmetric unification algorithm 
for one-sided distributivity. 
\end{abstract}

\maketitle
\section{Introduction}
Equational unification has long been a core component of automated deduction and 
more recently has found application in symbolic cryptographic protocol analysis~\cite{Escobar06}.  
In particular, the algorithm for 
unification modulo a \textit{one-sided} distributivity axiom \[ X \times (Y + Z) = X
\times Y + X \times Z \] is an early result by Tid\'{e}n and 
Arnborg~\cite{TidenArnborg87}.  
More recently this theory has been of interest in protocol analysis 
due to the fact that many cryptographic operators satisfy this property. 
Unfortunately the algorithm presented in the paper, although elegant and correct, 
has recently been shown not 
to be polynomial time bounded as claimed~\cite{narendran10}. 
In  addition, for some instances,  
there exist most general unifiers (mgus) that are exponentially large 
with respect to the input size.
In this paper we examine the \emph{decision} problem
for one-sided distributivity. 
More formally we  consider the decision problem 
for {\em elementary\/} unification modulo this theory,
where the terms can only
contain symbols in the signature of the theory and variables.
This is the theory considered by Tid\'{e}n and Arnborg~\cite{TidenArnborg87}. 
We first present a new polynomial 
time algorithm which solves the \emph{decision} problem for
a non-trivial subcase, based on a typed theory, of 
unification modulo one-sided distributivity.
This subcase happens to be sufficient to express the negative 
complexity result in~\cite{narendran10}.
Next we present a new polynomial algorithm which
solves the \emph{decision} problem for unification modulo 
one-sided distributivity. 
We employ string compression through the use of straight line programs,
which allows us to achieve the polynomial bound. 
Compression by straight line programs proves to be sufficient for our results, however
the use of compression in unification and matching is not novel to this paper. See 
for example~\cite{gascon11} and ~\cite{Levy08} for some pioneering work
on using compression in unification and other related problems. 

Since our initial results~\cite{marshall12c}, a new unification paradigm has been developed in~\cite{Erbatur13a} and is based on newly identified requirements arising
from the symbolic analysis of cryptographic protocols. In order to satify these requirements
and to apply state space reduction techniques, it is usually necessary for at least part of this state to be in normal form, and  to remain in normal form even after unification is performed.  This requirement can be expressed as an {\em asymmetric} unification
problem $\{ s_1 =^{\downarrow} t_1, ~\ldots, ~s_n =^{\downarrow} t_n \} $ where the $=^{\downarrow}$ denotes a unification problem with the  restriction that any unifier leaves the right-hand side of each equation irreducible. Given our motivation in protocol analysis,
our final result is to consider the theory in the newly developed paradigm and give the first asymmetric unification algorithm  for one-sided distributivity. 

\ignore{
\section{Contribution}
We present several new algorithms that tackle various unification problems
for the theory of one-sided distributivity. 
We present a new polynomial algorithm for a restricted form of the one-sided
distributivity problem. This problem, although a sub case of the general 
problem, is still sufficiently complex that the exponential problem
example is still a valid problem. In addition, the sub-problem provides
a nice illustration, minus the complex compression, of the method used
for the general problem. We present a new polynomial time bounded algorithm
for the unification decision problem for the theory one-sided 
distributivity. This algorithm employs string compression to achieve the result. Lastly we present a new algorithm for the asymmetric unification
problem in the theory of one-sided distributivity. The new
algorithm is not polynomial and the compression method does not directly
transfer from the unification to the asymmetric unification case the 
algorithm. However, the new algorithm is still interesting since it is the
first such algorithm, it is a easily understandable and implementable algorithm (much as the original Tid\'{e}n and Arnborg~\cite{TidenArnborg87}
 algorithm) and runs quickly on most cases.  
}

\section{Paper Outline}

Let us give a brief preview of the remaining portions of the paper. 
\begin{itemize}
\item Section~\ref{prelims} presents the preliminary background material.
\item Section~\ref{TA_Section} presents an overview of the complexity 
result concerning the original Tid\'{e}n and Arnborg~\cite{TidenArnborg87}
algorithm. 
\item Section~\ref{single_hom} presents the first contribution of this paper. We consider a restricted version of the one-sided distributivity 
problem, which is still sufficiently expressive to contain the family
of problems presented in Section~\ref{TA_Section}. For this new restricted 
version of the problem we develop a new polynomial time bounded 
decision algorithm (Algorithm~\ref{alg_A_2}). This section also provides an introduction to the 
methods used to solve the main problem. The solution to the main problem builds on Algorithm~\ref{alg_A_2} primarily by the addition of string compression.
\item Section~\ref{general_section} contains the main contribution of this
paper. Here we present a new polynomial bounded algorithm 
(Algorithm~\ref{alg_A_1}) for the decision unification problem over a 
theory of one-sided distributivity. The result is achieved by building
on Algorithm~\ref{alg_A_2} and using polynomial methods for
solving several problems on compressed strings. 
\item Section~\ref{asym_unif_sec} considers the one-sided distributivity 
problem in the new asymmetric unification paradigm. This new 
paradigm has only recently been identified (see~\cite{Erbatur13a}) and
is important to the area of cryptographic protocol analysis. Here
we present the first asymmetric unification algorithm for the theory
of one-sided distributivity. Although the algorithm is not polynomial
bounded, it should (much like the original Tid\'{e}n and Arnborg algorithm)
perform well on most problems. In addition, the algorithm is relatively 
simple, consisting of a small set of inference rules 
(see Figure~\ref{Asym_Inf_Rules}). One could ask why Algorithm~\ref{alg_A_1}
is not extended to the asymmetric problem? Unfortunately, the string compression methods required for the polynomial result in 
Section~\ref{general_section} do not easily extend to encapsulate the
additional information needed in an asymmetric unification problem.
This remains an open problem.      
\end{itemize}

\section{Preliminaries and General Results}\label{prelims} 

We use the standard notation of equational unification~\cite{BaaderSnyder}
and term rewriting systems~\cite{Baader98}.
The set of $\Sigma$-terms, denoted by $T(\Sigma, \mathcal{X})$, is built over the signature $\Sigma$ and the
(countably infinite) set of variables $\mathcal{X}$.  The terms $t|_p$ and $t[u]_p$ denote respectively the subterm of $t$ at the
position $p$, and the term $t$ having $u$ as subterm at position $p$. The symbol of $t$ occurring at the
position $p$ (resp. the top symbol of $t$) is written $t(p)$
(resp.  $t(\epsilon)$). 
The set of positions of a term $t$ is denoted by $Pos(t)$, the set of non variable positions for a term $t$ over a signature
$\Sigma$ is denoted by $Pos(t)_{\Sigma}$. 
A $\Sigma$-rooted term is a term whose top
symbol is in $\Sigma$.  The set of variables of a term $t$ is denoted
by $Var(t)$. 

A $\Sigma$-substitution $\theta$ is an endomorphism of $T(\Sigma,\mathcal{X})$ denoted by $\{ X_1 \mapsto t_1, \dots , X_n
\mapsto t_n \}$ if there are only finitely many variables
$X_1,\dots,X_n$ not mapped to themselves. We call domain of
$\theta$ the set of variables $\{ X_1,\dots, X_n\}$ 
and range of $\theta$ the set of terms $\{t_1,\dots,t_n\}$.  Application of a substitution $\theta$ to a term
$t$ (resp. a substitution $\phi$) may be written $t\theta$ (resp.~$\phi\theta$) or in functional notation as $\theta(t)$.

Given a first-order signature $\Sigma$, and a set $E$ of
$\Sigma$-axioms (i.e., pairs of $\Sigma$-terms, denoted by $l = r$),
the {\it equational theory} $=_E$ is the congruence closure of $E$
under the law of substitutivity. By a slight abuse
of terminology, $E$ will be often called an equational theory. 
\ignore{
An axiom $l=r$ is \emph{variable-preserving} if $Var(l) = Var(r)$. 
An axiom $l=r$ is \emph{linear} (resp. \emph{collapse-free}) if $l$ and $r$ are linear (resp. non-variable terms).
An equational theory is \emph{variable-preserving} (resp. linear/collapse-free) 
if all its axioms are variable-preserving (resp. linear/collapse-free).
An equational theory $E$ is \emph{finite} if for each term $t$, there are finitely many terms $s$ such that $t=_E s$.
}

Given an equational theory $E$, an $E$-unification problem 
is a set of equations \[ \mathcal{S} = \{ s_{1} =^{?} t_{1} , \dots , s_{m} =^{?} t_{m}\} \] 
A solution to $\mathcal{S}$, called an \emph{E-unifier\/}, is a 
substitution $\delta$ such that
$\delta(s_i) =_E^{} \delta(t_i)$ for all $1 \leq i \leq m$.
A substitution $\delta$ is \emph{more general
modulo\/} $E$ than $\theta$ on a set of variables $V$, 
denoted as $\delta \leq_{E}^V \theta$,
if and only if there is a substitution $\tau$ such that $\delta\tau(X)
=_{E} \theta(X)$ for all $X \in V$. 
Two substitutions $\theta_1^{}$ and $\theta_2^{}$ are 
{\em equivalent modulo\/} $E$ on a set of variables $V$, 
denoted as $\theta_1^{} \equiv_{E}^V \theta_2^{}$,
if and only if 
$\theta_1^{} (X) =_{E} \theta_2^{} (X)$ for all $X \in V$. 
For a substitution~$\theta$ and a set
of variables~$V$, $\theta|_V^{}$ denotes the restriction of the substitution
to the variables in~$V$, i.e., \[ \theta|_V^{} ~ = ~
\{ \, X \mapsto \theta(X) ~ | ~ X \in V \, \} \]

We call a set $\Gamma$ of substitutions a \emph{complete set of E-unifiers\/}
of $\mathcal{S}$ if and only if (i)  
for every $\theta \in \Gamma$, $\theta$ is an $E$-unifier 
and (ii) for every $E$-unifier $\theta$, 
there is a substitution $\delta \in \Gamma$ where $\delta \leq_{E}^{Var(S)} \theta$ holds.
A complete set of $E$-unifiers $\Gamma$
of a unification problem $\mathcal{S}$ is \emph{minimal\/} 
if and only if
for any two $E$-unifiers $\delta$ and $\theta$ in $\Gamma$, 
$\delta \leq_{E}^{Var(S)} \theta$ implies that $\delta = \theta$.

Equational unification problems are  classified based on the
function symbols that appear in them, i.e.,
their signature (\textit{Sig}). 
An $E$-unification problem~$S$
is {\em elementary\/} if and only if $Sig (S) \, = \, Sig(E)$. $S$ is
called an $E$-unification problem {\em with constants\/} if \mbox{$Sig (S)
\, \smallsetminus \, Sig(E)$} contains only free constants. Finally, if
there are uninterpreted function symbols in \mbox{$Sig (S) \, 
\smallsetminus \, Sig(E)$}, $S$ is called a general $E$-unification
problem. 

A set of equations $S$ is said to be in {\em standard form\/} 
over a signature $F$ if and only if
every equation in $S$ is of the form $ X =_{}^? t $ where $X$ is a variable
and $t$, a term over $F$, is one of the following: (a) a variable different from $X$, 
(b) a constant, or (c) a term of depth~1 that contains no constants.
We say $S$ is {\em in standard form\/} if and only if it is in standard form over the entire
signature. For a set of equations $S$ in standard form, ${lhs}(S)$
denotes the set of left-hand sides of equations in $S$. 
It is not generally difficult 
to decompose equations of a given problem into simpler standard forms.

A set of equations is said to be in {\em dag-solved form\/} if and only if they can be arranged as a 
list  \[ X_1 =_{}^? t_1 , \; \ldots , \; X_n =_{}^? t_n \] where
  (a) each left-hand side $X_i$ is a distinct variable, and
  (b) $\forall \, 1 \le i \le j \le n$: $X_i$ does not occur in
$t_j$~(\cite{JouannaudKirchner91}). A set of equations~$S$ is said to be in {\em $F$-solved form\/}
if and only if it is in standard form and the subset of equations
$S \, \cap \, (V \times T(F, V))$ is in dag-solved form.
Note, a unification problem in dag solved form has a unique most general idempotent unifier 
(see~\cite{JouannaudKirchner91}).

An equation $l = r$ is called a \emph{subterm collapsing} equation iff one term is a proper subterm of the
other. An equational theory, $E$, is called \emph{simple} or \emph{subterm collapse free} if there is no equation 
in $E$ that is subterm collapsing: $t \neq_{E} s$ for all proper subterms $s$ of $t$. Note, an important property
of simple theories is: a variable $X$ and a term $t$ are not $E$-unifiable if $X \in Var(t)$ (see~\cite{Burckert89}).   

\begin{definition}\label{slp_def}
A straight-line program ($SLP$) is a context-free grammar, 
$G=(\Sigma , N, P)$.
Where $\Sigma$ is the set of terminal symbols (these will correspond to
a set of ``label'' variables in this paper),
$N$ is a set of nonterminal symbols and $P$ as set of grammar productions.
$P$ contains only two types of productions: $N_i \rightarrow a$ and $N_i \rightarrow N_j
N_k$ with $i > j,k$, where $N_i^{}$, $N_j^{}$, $N_k^{}$ are
nonterminals and $a$ is a terminal. The SLP generates exactly one
string corresponding to the top nonterminal.
\end{definition}

As an example consider the string $(ab)^{32}$
\[abababababababababababababababababababababababababababababababab\] 
over the set of terminals $\left\lbrace a, b \right\rbrace $.
A corresponding $SLP$ for this string is [ $N_1 \rightarrow a$, $N_2 \rightarrow b$, $N_3 \rightarrow N_1 N_2$,
$N_4 \rightarrow N_3 N_3$, $N_5 \rightarrow N_4 N_4$, $N_6 \rightarrow N_5 N_5$, $N_7 \rightarrow N_6 N_6$,
$N_8 \rightarrow N_7 N_7$ ].
The size of a $SLP$ can be defined in several ways. We use the following definition from
~\cite{levy04}. We note that the name $SLP$ is not used in~\cite{levy04}
and \cite{Levy08} rather they use the name Singleton Context-Free Grammars.
For any terminal, $a$, define the $depth(a) = 0$ and for any nonterminal $N$ 
\[depth(N) = max_{\{k=1,2\}}\left\lbrace 
depth(N_{k}) + 1 ~|~ N \rightarrow N_{1}N_{2} 
\right\rbrace \]
We can define the depth of the $SLP$ as the depth of its top nonterminal.
The \textit{size} of a $SLP$, $S$, is defined to be the number of productions and is denoted as $|S|$.
We denote the length of a string produced by a $SLP$ $S$ by $\|S\|$, $64$ in the above example. Note that the
lengths are represented in \emph{binary}.


\section{The Tid\'{e}n-Arnborg Algorithm and the Exponential Examples}\label{TA_Section}
Here we wish to very briefly 
review the Tid\'{e}n-Arnbog Algorithm~\cite{TidenArnborg87} 
and the exponential time result~\cite{narendran10}.
We can assume, without loss of generality, that the input
is given as a set of equations, where each equation is in one of
the following forms: \[ X =_{}^? Y, ~ X =_{}^? Y + Z, ~
\mathrm{and} ~ X =_{}^? Y \times Z \]
   
A simple decomposition algorithm can transform a set of equations into
the above form and maintain unifiability (see~\cite{TidenArnborg87}).

\subsection{The Tid\'{e}n-Arnborg Algorithm} In~\cite{TidenArnborg87} 
Tid\'{e}n and Arnborg developed an elegant  algorithm which is based
on the following results.
\begin{theorem}\label{TA_sound}
\textit{(Tid\'{e}n and Arnborg~\cite{TidenArnborg87})}\\
In the theory of one-sided distributivity:
 \begin{enumerate}
 \item The set of equations $\{ U =^{?} X_1 \circ X_2, ~U =^{?} T_1 \circ T_2 \}$
has precisely the same unifiers as the set of equations 
$\{ U =^{?} X_1 \circ X_2,  ~X_1 =^{?} T_1, ~X_2 =^{?} T_2\}$,
where $\circ$ is $\times$ or $+$.
\item  Every unifier for the set of equations\\ 
$\qquad \{ U =_{}^? V \times W, \; W =_{}^? W_1^{} + W_2^{}, \; X =_{}^? V \times W_1^{}, \; Y =_{}^? V \times W_2^{} \}$ \\ 
is a unifier for the set of equations
$\{ U =_{}^? V \times W, \; U =_{}^? X + Y \}$, where 
$W_1$ and $W_2$ are fresh variables. 
 \end{enumerate}
\end{theorem}
 
\noindent The key steps in the algorithm
can be described by the deduction rules of Figure~\ref{TA_Inf_Rules}. 

\begin{figure}[ht]
  \centering
\scalebox{.95}{
\fbox{
\begin{small}
\begin{tabular}{lcc}
(a) & $\qquad$ & $\vcenter{
\infer[\qquad \mathrm{if} ~ U ~ \mathrm{occurs ~ in} ~ \eq ]{\{U =_{}^? V\} \cup \,  \{ U \mapsto V \}(\eq) }
      { \{U =_{}^? V\} ~ \uplus ~ \eq }
}
$\\[+20pt]
(b) & & $\vcenter{
\infer{\eq ~ \cup ~ \{ U =_{}^? V \times W, \; V =_{}^? X, \; W =_{}^? Y \}}
{\eq ~ \uplus ~ \{ U =_{}^? V \times W, \; U =_{}^? X \times Y \}}
}
$\\[+20pt]
(c) & & $\vcenter{
\infer{\eq ~ \cup ~ \{ U =_{}^? V + W, \; V =_{}^? X, \; W =_{}^? Y \}}
{\eq ~ \uplus ~ \{ U =_{}^? V + W, \; U =_{}^? X + Y \}}
}
$\\[+20pt]
(d) & & $\vcenter{
\infer{\eq ~ \cup ~ \{ U =_{}^? V \times W, \; W =_{}^? W_1^{} + W_2^{}, \; X =_{}^? V \times W_1^{}, \; Y =_{}^? V \times W_2^{} \}}
{\eq ~ \uplus ~ \{ U =_{}^? V \times W, \; U =_{}^? X + Y \}}
}
$
\end{tabular}
\end{small}
} }
 \caption{Tid\'{e}n and Arnborg Inference Rules}
  \label{TA_Inf_Rules}
 \end{figure} 

The $W_1, W_2$ in rule~(d) are fresh variables and $\uplus$ is
{\em disjoint union.\/} Furthermore, rule~(d) (the ``splitting
rule'') is applied only when the other rules cannot be applied. A
set of equations is said to be {\em abc-reduced\/}\footnote{Such a system
of equations is called \emph{simple} in ~\cite{TidenArnborg87}. However,
simple has come to denote a theory that is subterm collapse free 
(see~\cite{Burckert89}).}
 if and only if none
of the rules (a), (b) and (c) can be applied to it.  A {\em sum
transformation\/} is defined as a binary relation between two
abc-reduced systems, $S_1^{}$ and $S_2^{}$, where $S_2^{}$ is obtained
from $S_1^{}$ by applying rule~(d), followed by exhaustive
applications of rules (a), (b) and (c). Clearly, a sum
transformation is applicable if and only if some variable occurs
as the left-hand side in more than one equation.  

The algorithm also makes use of two graphs. The graphs are used to detect
two types on non-unifiability errors. We include the definitions here for completeness. 

\begin{definition}\label{orig_DG}
 The {\em dependency graph\/} ($D(S)$) of an abc-reduced system, $S$, is an 
edge labeled, directed multi-graph.
It has as vertices the variables of $S$. For an equation $X= Y + Z$ in $S$
it has an $l_{+}$-labeled edge $(X,Y)$ and an $r_{+}$-labeled edge $(X,Z)$. An equation
$X = Y \times Z$ similarly generates two edges with labels $l_{\times}$ and $r_{\times}$.
\end{definition}

\begin{definition}\label{orig_PL}
The {\em sum propagation graph\/} ($P(S)$) of an abc-reduced system $S$ is a
directed simple graph.  It has as vertices the equivalence classes of
the symmetric, reflexive, and transitive closure of the relation
defined by the $r_{\times}$-edges in the dependency graph of
$S$.  It has an edge $(V,W)$ iff there is an edge in the
dependency graph from a vertex in $V$ to a vertex in $W$ with label
$l_{+}$ or $r_{+}$.
\end{definition}

It can be seen that by using cycle checking on $D(S)$ we can detect 
all the occur-check like errors that may develop as the algorithm works
with the system of equations. We know these are indeed errors due to the following
property.

\begin{theorem}\label{TA_cycle}
The one-sided distributive axiom is subterm-collapse free.
\end{theorem}
\begin{proof}
 If we consider the convergent system 
\[ X \times (Y + Z) \rightarrow X \times Y + X \times Z \]
we can see that the rule is non-size-reducing.
Therefore, we cannot reduce a term $t$ to a subterm of itself.  
\end{proof}

This implies that the system is simple (\cite{Burckert89}) and therefore occur-checks must be detected as they imply non-unifiability.

The propagation graph is needed to detect non-unifiable systems that
cause infinitely many applications of the splitting rule~(d).
An
example of this type of system is the following two
equations:
\begin{center}
 $Z =^? V_2 + V_3$,  $Z =^? V_1 \times V_3$.
\end{center}
These types of systems are shown not to have a finite unifier~\cite{TidenArnborg87}. However,  they will never produce a cycle in the dependency graph, 
thus the propagation graph is needed.

We can conclude this overview of the original algorithm with some of the 
results proven for it in~\cite{TidenArnborg87}:
\begin{theorem}\label{TA_Results}
\textit{From Tid\'{e}n and Arnborg~\cite{TidenArnborg87}:}
\begin{enumerate}
\item The algorithm formed by applying the \emph{sum transformation}
with the rules of Figure~\ref{TA_Inf_Rules} is sound, complete and terminating. 
\item If the system is not unifiable either the dependency graph
(Definition~\ref{orig_DG}) or the propagation graph 
(Definition~\ref{orig_PL}) will contain a cycle after a finite number of steps. 
\item If either the dependency or the propagation graph contain a cycle,
the initial system is not unifiable. 
\item The algorithm produces a final solved form, which provides a unique most general unifier for the initial system. 
\end{enumerate}
\end{theorem}

\subsection{Complexity Result}
In~\cite{narendran10} a family of \textit{unifiable}, \emph{abc-reduced} systems
is presented, on which the Tid\'{e}n-Arnborg
algorithm runs in exponential time. 

\begin{definition} \label{EQ} \cite{narendran10}
 Let $EQ$ be a subset of the set of abc-reduced systems defined as follows: all
 multiplications are of the form $X_{i} \ueq T \times Y_{j}$ (or
 $Y_{j}  \ueq  T \times X_{i}$) where $T$ is a unique variable and all
 additions are of the form $X_{i} \ueq X_{i1}+X_{i2}$ or
 $Y_{i} \ueq Y_{i1}+Y_{i2}$.
\end{definition}
That is variables are represented using $X$ and $Y$ along with subscripts.
The actual family of instances that causes the exponential growth is 
a subset of $EQ$ defined as: 
\begin{definition} \label{sigman} \cite{narendran10}
For $n \ge 0$, let $\sigma(n)$ be the set of equations
\begin{eqnarray*}
X_{1^{i}} &  \ueq  & X_{1^{i+1}}+X_{1^{i}2}, \\
Y_{2^{i}} &  \ueq  & Y_{2^{i}1}+Y_{2^{i+1}}, \\
Y_{2^{i}1} &  \ueq  & T \times X_{1^{i}2},\\
X &  \ueq  & T \times Y, \\
X_{1^{i+1}} &  \ueq  & X_{1^{i+2}}+X_{1^{i+1}2}
\end{eqnarray*}
for all $0 \le i \le n$. Where $X_{l^{i}}$ denotes $i$ concatenations of $l \in \{1,2\}$, i.e.,
$X_{1^{3}2} = X_{1112}$. 
\end{definition}
It is shown in~\cite{narendran10} that a system of equations, as defined in
Definition~\ref{sigman}, will result in exponentially many applications of
the sum transformation. 

The result can be viewed graphically in the following manner.  Let
variables represent nodes in a graph and create downward edges for
variables related by an addition operation and lateral edges for
variables related by a multiplication operation (this is essentially the $D(S)$ definition). 
The edges to the unique variable $T$ will not affect the complexity and so can be
ignored.  We can see such graphs in the following two examples.
Figure~\ref{fig1} represents an initial set of equations and
Figure~\ref{fig2} is the same system after application of the
Tid\'{e}n-Arnborg algorithm.  Essentially, we see that the new variables
and paths that are created at each level of the graph are the cause of
the complexity growth and will need to be avoided.

\begin{figure}[H]
  \centering
   \scalebox{.96}{
   \input{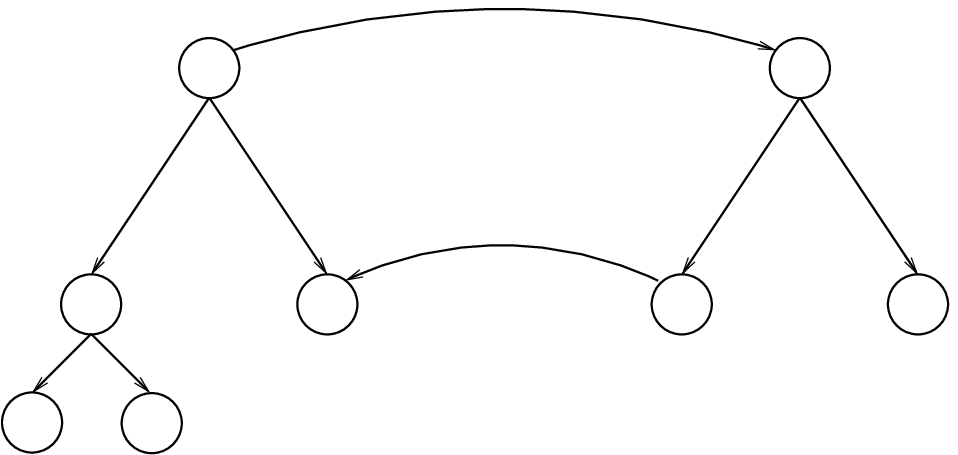} 
   }
  \caption{Graph for $\sigma(0)$} \label{fig1}
\end{figure}

\begin{figure}[H]
  \centering
   \scalebox{.96}{
   \input{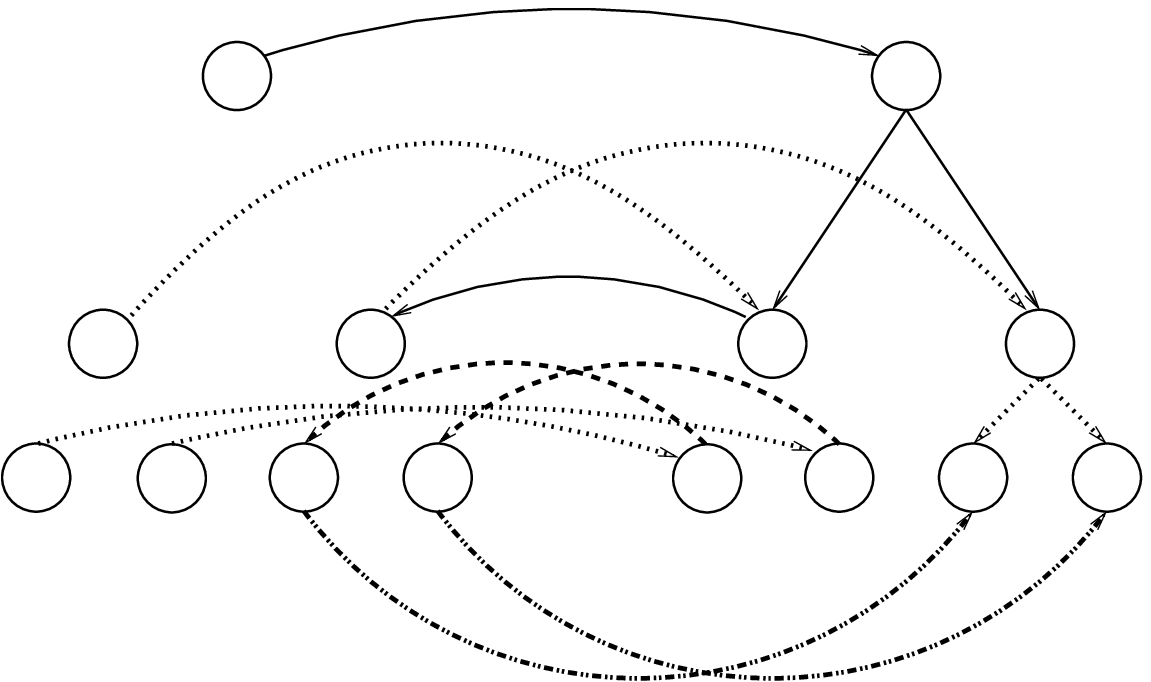}
   }
  \caption{After 4 applications of the sum transformation} \label{fig2}
\end{figure}

In the Tid\'{e}n-Arnborg algorithm this exponential behavior is due to
an exponential number of application of rule (d) from Figure~\ref{TA_Inf_Rules}. This is the rule that creates the new variables and
paths seen in Figure~\ref{fig2}. We develop a new algorithm in 
Section~\ref{single_hom} which ensures a polynomial number of application
of a rule equivalent to rule (d) from Figure~\ref{TA_Inf_Rules} and this
algorithm is sufficient to ensure polynomial time and solve the unification
problem for the Single Homomorphism (introduced in the next section) 
restricted form of the problem. However, when applied to the full
problem it proves insufficient, see example~\ref{string_length_ex}.
The solution is to introduce the use of string compression, which
is done in Section~\ref{general_section}.     


\section{Typed System and Single Homomorphism}\label{single_hom}
We present a typed system interpretation of one-sided distributive
unification.  We begin with the simplest non-trivial subcase, the case
of a \textit{single homomorphism}.  This is non-trivial because the
exponential complexity result in~\cite{narendran10} holds in this case
as well. Consider a `type' system based on two types $\tau_1$ and
$\tau_2$. We
let all left multiplication variables be of type $\tau_1$ and all right 
variables of type $\tau_2$. Thus \[
\begin{array}{lccc}
\times : & \tau_1 * \tau_2 &  \rightarrow  & \tau_2, \\
+: & \tau_2 * \tau_2 &  \rightarrow  & \tau_2, 
\end{array} \] 
If there is only a single variable of type $\tau_1$ in the input
equations then we can consider the multiplication operation as a
homomorphism $h$ over $+$. Thus, we can view an equation of the form
$X = T \times Y $, where $T$ is the single variable of type $\tau_1$, as
the homomorphism equation $X = h(Y)$.
This is the single homomorphism case, 
it restricts the number of valid terms from the general case but 
it is sufficient for encoding the exponential example
in~\cite{narendran10} and it yields a much simplified decision
algorithm.
\subsubsection{Single homorphism and the General Algorithm}
There are two primary reasons for considering this sub-case:
\begin{enumerate}
\item The Algorithm for the single homomorphism case (Algorithm~\ref{alg_A_2}) is more efficient than the algorithm 
that solves the general case (Algorithm~\ref{alg_A_1}).\\
This is due to the $SLPs$. In Algorithm~\ref{alg_A_1}, every step dealing with compression must use $SLPs$ and thus
must employ various subroutines for dealing with $SLPs$, of which the best
complexity measures are all of quadratic or greater polynomial complexity. 
However, in this restricted case binary encoding provides suitable compression. Operations dealing with compressed
objects are reduced to addition and subtraction, i.e., linear complexity. 
\item Algorithm~\ref{alg_A_1} is built from Algorithm~\ref{alg_A_2}.\\
Algorithm~\ref{alg_A_2} uses the same underlying method
used in Algorithm~\ref{alg_A_1}. 
Both algorithms approach the problem by ordering the
equivalence classes (defined below) and ``processing'' each class, one at a time. However, the processing is less complex in this restricted theory
since it does not need to deal with $SLPs$. 
This leads to a similar algorithm which is easier to
understand.
\end{enumerate}

\subsection{Data Structures}
\begin{definition}
We define the following relations ($X$, $Y$ and $Z$ are variables):

\begin{itemize}

\item  $X \succ_{h} Y$ if $X = h(Y)$.

\item  $X \succ_{l_+} Y$ if $X = Y + Z$.

\item  $X \succ_{r_+} Z$ if $X = Y + Z$. 

\item $X \succ_{a} Z$ if $X = Y + Z$ or $X = Z + Y$.
\end{itemize}
\end{definition}

We use the following two graphs, that are similar to the
dependency and propagation graphs used in~\cite{TidenArnborg87},
see Definitions~\ref{orig_DG} and~\ref{orig_PL}.
For a unification problem $S$ in standard form we construct the following
two graphs.
\begin{definition} \label{MD}
A {\em path labeled dependency graph\/} ($\mathcal{LD}$) is a
directed graph such that the nodes in the graph correspond to
variables of type $\tau_2$.  We form two kinds of edges:\\
(i) {\em Lateral} edges, where for each equation of the form $X =_{}^? h(Y)$, 
we have an edge from node $X$ to node $Y$ labeled with
a {\em label variable\/}, $h^{1}$. 
Thus, for single edges corresponding to a single homomorphism the label is $h^{1}$.
For paths corresponding to multiple homomorphisms
(compound paths which will be constructed during the running of the algorithm) 
the label is $h^{n}$, $n \in \mathbb{N}^{+}$, where
$n$ is the number of homomorphisms/single edges composing the path.
We will use $\pi$ (possible with subscripts) to denote a path in the graph.
For example, for a path between nodes $X$ and $Y$ we write,  
$X \; \stackrel{\pi}{\longrightarrow} \; Y$. Where $\pi$ is understood to represent
some $h^{j}$, $j \in \mathbb{N}$. 
The path length, denote $|\pi|$, is $j$. \\
(ii) {\em Downward} edges, where for
each equation of the form $X =_{}^? X_1^{} + X_2^{}$, we have
directed edges from node $X$ to node $X_1$ and from node $X$ to
node $X_2$.
\end{definition}

\begin{definition} \label{MP}
The {\em path labeled propagation graph} ($\mathcal{LP}$) is a
directed simple graph.  Its vertices are the equivalence classes
of the symmetric, reflexive, and transitive closure of the
relation defined by $\succ_{h}$ on the $\mathcal{LD}$ graph for
the same system.  Edges exist between equivalence classes
$\left[X \right]$ and $\left[ Y \right]$ if there exist
variables $U \in \left[ X \right] $ and $V \in \left[ Y \right] $
such that $U \succ_{a} V$. 
\end{definition} 

Note, that we can order the equivalence classes/nodes of $\mathcal{LP}$.
Let~$\sim_h^{}$ stand for the reflexive, symmetric and transitive
closure of~$\succ_{h}$. Thus $\sim_h^{}$ defines a set of equivalence
classes over a set of variables. Denote these classes as $\left[ Y \right]_h$.
The $\mathcal{LP}$ graph has exactly these classes as its nodes. 
We can define a strict partial ordering $\gtrdot_h$ on the $\sim_h^{}$-equivalence classes
based on $\succ_a^{}$. That is, $\left[ X \right]_h \gtrdot_h \left[ Y \right]_h$ 
if and only if there exist 
$ K_1 \in \left[ X \right]_h$ and $K_2 \in \left[ Y \right]_h$
such that $K_1 \succ_a^{} K_2$, i.e., an edge from the node $\left[ X \right]_h$ to the node $\left[ Y \right]_h$. 
This ordering  will be important as it provides an ordering strategy  for applying the unification algorithm. 

These graphs, mainly the $\mathcal{LD}$, 
will be the primary data structure and will be modified via the set of graph {\em saturation\/} rules. 
The rules are very similar to the original Tid\'{e}n-Arnborg rules however they primarily act not on the set of equations
but on the $\mathcal{LD}$ graph. This is due to the need for compression, where acting on a fully uncompressed
set of equations results in the original algorithm.   
Note, that we still need the $\mathcal{LP}$ graph for detecting the set of non-unifiable systems. 
An example of this is the following set of equations.
\[ \{
X =_{}^? V + Y, \;
X =_{}^? h(Y) \} \]

The $\mathcal{LP}$ graph and the sum propagation graph of~\cite{TidenArnborg87} 
(Definition~\ref{orig_PL}) are the same for the single homomorphism systems.  
This is easy
to see as both graphs will contain the same equivalence classes and thus 
nodes and both graphs have the same edges. Therefore, each time the algorithm 
updates the $\mathcal{LD}$ graph (i.e., the inference rules modify the $\mathcal{LD}$ graph) 
it also updates the $\mathcal{LP}$ graph 
and checks for cycles. Likewise, if cycles are found the algorithm terminates with failure.  
\subsection{Algorithm Presentation}
Before presenting the  rules, we need to discuss several problems 
the algorithm needs to solve when dealing with compressed paths. 
During saturation we derive {\em path constraints\/} of 
the form
$\pi_1^{} =_{}^? \pi_2^{}$ or $\pi_1^{} {\prec}_{}^? \pi_2^{}$. For the single
homomorphism case, because there is just one homomorphism, $\pi_1^{} =_{}^? \pi_2^{}$
is simply a check if the lengths are equal, i.e., if $| \pi_1 | = | \pi_2 |$. 
For
the prefix check $\pi_1^{} {\prec}_{}^? \pi_2^{}$,
in the single homomorphism case we only need to check if the length
of $\pi_1$ is less then $\pi_2$, i.e., $| \pi_1 | < | \pi_2 | $. 
It is important to note that path lengths are kept in binary
representation.  This compression is significant as it allows us to
avoid exponential growth in the path lengths. In addition to path
constraints we will need to perform several {\em path computations\/},
specifically we need to {\em concatenate paths\/} and 
{\em compute path suffixes\/}. These operations can be accomplished, in the
single homomorphism case, by simple addition and subtraction.

We now introduce a set of inference rules.
\emph{Rule $(0)$ acts on the system $S$ and rules $(i)$ through $(vii)$
act on the $\mathcal{LD}$ graph of $S$}.
Rule $(0)$ is simple variable replacement.
Rules $(i)$ - $(iii)$ are cancellation rules that follow 
directly from the rules of Figure~\ref{TA_Inf_Rules}.
Rule $(vi)$ is a failure rule 
that corresponds to occur-check type errors. 
Rules $(iv)$, $(v)$, $(vii)$ are path completion rules. 
Rule $(vii)$ is the same path propagation rule 
from the Tid\'{e}n-Arnborg algorithm, 
justified by the axioms of the system; 
see Figure~\ref{TA_Inf_Rules} rule(d).
\emph{However, in rule $(vii)$ we do not create 
the new variables $W_1$ and $W_2$
unless $W$ has no child variables related along a $\succ_a$ edge}.

\begin{figure}[ht]
  \centering
\scalebox{.82}{
\fbox{
\begin{tabular}{lcc}
$(0)$ & $\qquad$ & $\vcenter{
\infer[\qquad \mathrm{if} ~ U ~ \mathrm{occurs ~ in} ~ S ]{ \{ U \mapsto V \}(S) \cup 
\, \{U =_{}^? V\} }
      { S ~ \uplus ~ \{U =_{}^? V\} }
}
$\\[+20pt]
$(i)$ & & $\vcenter{
\infer{U=_{}^? U_1 + U_2 , \; U_3 = U_1 , \; U_4 = U_2}
{ U=_{}^? U_1 + U_2 , \; U =_{}^? U_3 + U_4, \; } }
$\\ [+20pt] 
$(ii)$ & $\qquad$ & $\vcenter{
\infer{X \, {=}_{}^? \, Z, \; \qquad Z \; \stackrel{\pi}{\longrightarrow} \; Y}
{X \; \stackrel{\eta}{\longrightarrow} \; Y, \; \qquad Z \; \stackrel{\pi}{\longrightarrow} \; Y \; \qquad | \eta | \, = \, | \pi |}
}
$\\[+20pt]
$(iii)$ & $\qquad$ & $\vcenter{
\infer{X \; \stackrel{\pi}{\longrightarrow} \; Z, \; \qquad Y \, {=}_{}^? \, Z}
{X \; \stackrel{\eta}{\longrightarrow} \; Y, \; \qquad X \; \stackrel{\pi}{\longrightarrow} \; Z, \; \qquad | \eta | \, = \, | \pi |}
}
$\\[+20pt]
$(iv)$ & $\qquad$ & $\vcenter{
\infer{Y \; \stackrel{h^{i-j}}{\longrightarrow} \; Z, \; \qquad X \; \stackrel{h^{i}}{\longrightarrow} \; Z}
{X \; \stackrel{h^{j}}{\longrightarrow} \; Y, \; \qquad X \; \stackrel{h^{i}}{\longrightarrow} \; Z, \; \qquad j \, < \, i}
}
$\\[+20pt]
$(v)$ & $\qquad$ & $\vcenter{
\infer{X \; \stackrel{h^{i+j}}{\longrightarrow} \; Z, \; \qquad Y \; \stackrel{h^{j}}{\longrightarrow} \; Z}
{X \; \stackrel{h^{i}}{\longrightarrow} \; Y, \; Y \stackrel{h^{j}}{\longrightarrow} \; Z }
}
$\\[+20pt]
$(vi)$ & $\qquad$ & $\vcenter{
\infer[\qquad \mathrm{if} ~ \mathcal{LP} ~\mathrm{or} ~ \mathcal{LD} ~\mathrm{are ~cyclic} ]{FAIL}
{S}
}
$\\[+20pt]
$(vii)$ & & $\vcenter{
\infer{U \; \stackrel{\eta}{\longrightarrow} \; \; W, \; 
W =_{}^? W_1^{} + W_2^{}, \; U_1 \; \stackrel{\eta}{\longrightarrow} \; W_1^{},
 \; U_2 \; \stackrel{\eta}{\longrightarrow} \; W_2^{} }
{ U \; \stackrel{\eta}{\longrightarrow} \; \; W, \; U =_{}^? U_1 + U_2, \; } }
$\\ 
\end{tabular}
} }
 \caption{Inference Rules for the Single Homomorphism Problem.}
  \label{Typed_Inf_Rules}
\end{figure}

Before giving the algorithm details let us give a high-level
overview of the process.
\begin{enumerate}
 \item The algorithm begins with a unification problem, $S$, in standard form.
\item From the set of equations $S$ it generates the
$\mathcal{LD}$ graph and from the $\mathcal{LD}$ graph it generates the $\mathcal{LP}$ graph.
\item Next the algorithm applies the set of ``cancellation'' inference rules. These are rules which do not create new edges
or nodes in the graph and clearly terminate. 
\item The algorithm works in a top down ordering on the equivalence classes, using the relation $\gtrdot_h$
to order the classes. Each class is ``processed'' using the inference rules. This is done by applying the rules
to the nodes in $\mathcal{LD}$ graph which are contained in the current ``selected'' class.
\item After each new class is processed  the algorithm applies the cancellation rules and re-checks for any errors.
\item During this process two things can happen:
\begin{enumerate}
 \item Cycles can be found in either graph implying non-unifiability.
 \item The inference rules are exhaustively applied and no cycles occur, implying unifiability.  
\end{enumerate}
\end{enumerate}

\noindent The algorithm for the single homomorphism subcase is presented in Algorithm~\ref{alg_A_2}. 
 
\begin{algorithm}
\caption{Unification modulo a Single Homomorphism}
\label{alg_A_2}
\begin{algorithmic}
\vspace{0.1in}
\STATE (Input: A system of equations in standard form)
\STATE \textbf{(1: Generate data structures)} Generate the graphs, $\mathcal{LD}$ and $\mathcal{LP}$. 

\STATE \textbf{(2: Clean up the system)} Exhaustively apply the rules
$(0)$, $(i)$, $(ii)$, $(iii)$ and $(iv)$. 

\STATE \textbf{(3: Error checking)} 
Apply graph cycle checking to the two graphs (i.e., rule $(vi)$). 
If a cycle is found \textit{stop} with failure. 

\STATE \textbf{(4: Process equivalence class)}
Select an equivalence class based on the strict partial ordering $\gtrdot_h$.
That is, we select the largest element of $\gtrdot_h$ that
has not yet been processed. Thus, if we select the class
$\left[ X \right]_h$ then there does not exist a class $\left[ Y \right]_h$
such that $\left[ Y \right]_h$ has not been processed and 
$\left[ Y \right]_h$~$\gtrdot_h^{+}$~$\left[ X \right]_h$.
Clearly, if $\gtrdot_h$ is not a
strict partial ordering then there is a cycle in the $\mathcal{LP}$ graph.\\[+5pt]

First we apply rule $(v)$ --- this is done by starting with the
sink node of the path and working back to the start node of the
path.  Once rule~$(v)$ has been exhaustively applied we apply
rule~$(vii)$ if applicable.

\STATE \textbf{(5: Check if Complete)}
If no inference rules can be applied and no cycles exist,
then exit with success, else return to Step~2.
\vspace{0.1in}
\end{algorithmic}
\end{algorithm}

We next discuss the
correctness and complexity of Algorithm~\ref{alg_A_2}; most of these
results will follow directly from~\cite{TidenArnborg87}. 
 
\subsection{Correctness}
Correctness of the inference rules can be assured due to the 
correctness proof of the algorithm presented in~\cite{TidenArnborg87}
and the following lemmas.
\begin{lemma} \label{sh_sound}
 Soundness of rules $(i)$ through $(vii)$ are direct consequences of the 
``sum transformation''\footnote{See Section~\ref{TA_Section} for the definition of sum transformation. } 
method of ~\cite{TidenArnborg87} and variable replacement. 
\end{lemma}
\begin{proof}
The soundness of the rules follow from Theorem~\ref{TA_sound}.
Therefore, we know that the set of equations
$\{ X=^{?} Y \circ Z, ~X=^{?} V \circ W \}$, where $\circ$ is $+$ or $\times$, has the same solutions
as the set $\{ X =^{?} Y \circ Z, ~Y =^{?} V, ~Z=^{?} W \}$, that the set of equation 
$\{ X=^{?} Y \times Z, ~X=^{?} V + W \}$ has the same solutions, over the shared variables, as
the set $\{ X =^{?} Y \times Z, ~Z=^{?} V_1 + V_1,  ~W =^{?} Y \times V_1, ~V=^{?} Y \times V_2 \}$.
\end{proof}
The $\mathcal{LD}$, and $\mathcal{LP}$, graphs are
simply graphical representations of a system of equations, which 
Algorithm~\ref{alg_A_2} transforms by application of one or more of the inference rules. 
Lemma~\ref{sh_sound} ensures that each transformation is sound. 
It remains to be shown that if the algorithm terminates without 
failure then the system is indeed unifiable.

\begin{lemma}\label{dagsolved2}
Given a system of equations $S$ in standard form  
if no failure errors occur Algorithm~\ref{alg_A_2} transforms $S$, 
through its $\mathcal{LD}$ graph representation, into dag-solved form.  
\end{lemma}
\begin{proof}
Let $D$ be the final $\mathcal{LD}$ graph and consider 
the definition of dag-solved form. 
\begin{itemize}
\item The first condition is satisfied as each variable is
represented by a node in the graph. If the left hand sides~$X_i$
were not distinct, then a cancellation or path propagation rule, 
$(vii)$, could be applied. 
\item The second condition is satisfied as the paths correspond
to a distinct ordering and there are no cycles in the graph. \qedhere
\end{itemize}
\end{proof}

\noindent Therefore,  if the system is unifiable the algorithm
will report that fact. We need to show that if the system is not
unifiable the algorithm correctly reports that as a failure.
Directly from~\cite{TidenArnborg87} we get the following two results.

\begin{lemma} \label{LD_cycle}
Cycles in the $\mathcal{LD}$ graph for a system 
$S$ in standard form imply that $S$ is not unifiable.
\end{lemma}
\begin{proof}
This is due to Theorem~\ref{TA_cycle},
which shows that the one-sided distributive axiom is 
subterm-collapse free. The constraint to a typed system does not remove the property that the system is simple. Therefore, a cycle in the
$\mathcal{LD}$ graph will imply a cycle in the system of equations
and a non-unifiablity error for a simple system. 
\end{proof}

\begin{lemma} \label{LP_cycle}
Cycles in the $\mathcal{LP}$ graph for a system 
$S$, in standard form, imply that $S$ is not unifiable.
\end{lemma}
\begin{proof}
The $\mathcal{LP}$ graph contains the same information, 
for the single homomorphism systems,  as is
contained in the propagation graph of~\cite{TidenArnborg87} 
(Definition~\ref{orig_PL}). Both graphs will contain 
the same equivalence classes and thus nodes and both graphs 
have the same edges. The result then follows from 
Theorem~\ref{TA_Results}.  
\end{proof}

\begin{theorem}
Algorithm~\ref{alg_A_2} is correct.
\end{theorem}
\begin{proof}
Follows from Lemma~\ref{sh_sound} to Lemma~\ref{LD_cycle}.
\end{proof}

\subsection{Complexity}
First we get the following result from the cancellative nature of the 
rules $(i)$ through $(iii)$.
\begin{lemma}
Given a $\mathcal{LD}$ graph rules $(0)$-$(iii)$ can only be applied
a polynomial number of times with respect to the initial set of nodes
in the graph.
\end{lemma}
In addition, we get the following clear result.
\begin{lemma}
Given a $\mathcal{LD}$ graph rule $(iv)$ can only be applied
a polynomial number of times with respect to the initial set of nodes
in the graph.
\end{lemma}

\begin{lemma}\label{Single_Sink}
Each equivalence class formed by closure along $\succ_h$-related nodes
has a unique sink.
\end{lemma}
\begin{proof}
If a class has no sink then there is a cycle and the system is not unifiable.
Now assume we have at least one sink. 
Rules $(iii)$, $(ii)$, $(iv)$ and $(v)$ ensure
that each node can have at most one lateral outgoing edge.
\end{proof}

\begin{lemma}\label{processing}
Processing an equivalence class (Step 4)
takes polynomial time with 
respect to the number of variables in the class.
\end{lemma}
\begin{proof}
By rules $(i)$ through $(iv)$ each variable in the class
will have at most one outgoing edge and all paths will lead to the sink. 
Applying $(v)$ exhaustively starting from the sink is therefore
bounded linearly by the number of variables in the class. In addition,
$(vii)$ is also bounded by the number of variables in the class as it can be applied at
most once for each variable in the class. Moreover, it can create 2 new variables
at most once for each class. 
\end{proof}

\begin{lemma}\label{num_equiv_class}
The number of $\sim_h^{}$-equivalence classes for a system
$S$ can never increase.
\end{lemma}
\begin{proof}
New variables created by rule $(vii)$ don't create new equivalence
classes as they are added to pre-existing classes.
\end{proof}

Due to the fact that each equivalence class contains a single sink, 
by Lemma~\ref{Single_Sink}, we get the following.

\begin{lemma}\label{nodes_added}
A maximum of $2$ new nodes can be added to an equivalence class
from any one higher equivalence class.
\end{lemma}

In addition by rule $(vii)$ we get the following.

\begin{lemma}\label{paths_added}
During processing the number of paths added to a
$\sim_h^{}$-equivalence class from a higher, by $\gtrdot_h$,
$\sim_h^{}$-equivalence class cannot exceed the number of nodes
in the lower equivalence class.
\end{lemma}

Combining the above results we get the following.

\begin{theorem}
The running time of 
Algorithm~\ref{alg_A_2} is polynomial with respect to 
the initial set of equations.
\end{theorem}

\begin{proof} 
Processing an equivalence class is polynomial bounded by
Lemma~\ref{processing}. \\ 
By Lemma~\ref{nodes_added} and
Lemma~\ref{paths_added} the classes can only grow by a constant amount as
each class is processed and by Lemma~\ref{num_equiv_class} the number
of classes cannot increase.
\end{proof}

This section covers a decision algorithm for the single homomorphism
subcase.  The obvious extension to this problem results in the
{\em multiple homomorphism\/} problem.  In the multiple homomorphism
case we may have a finite set of variables of type $\tau_1$ but we can
still consider them as homomorphisms $h_1, \ldots, h_n$.  Although we
do not go into any more details here, the multiple homomorphism
case is also interesting. Unlike the single homomorphism case
compression is needed for the multiple homomorphism case.  This is due
to the fact that, unlike the single homomorphism case, the label
variables are not the same and therefore just keeping the path lengths
is not sufficient.  But, the multiple homomorphism case does not
require all the methods presented in the next section for the general case, due to
the type system, i.e., labeled variables cannot also be nodes 
in the $\mathcal{LD}$ graph.

\section{General Algorithm}\label{general_section}
We now consider the general problem, with no type system.
Let us give a brief overview of the section.
\subsubsection*{Section Summary}
\begin{itemize}
\item We begin with a discussion on why the string compression methods are
required to achieve a polynomial bound, Example~\ref{string_length_ex}.
\item We next introduce the new graph data structures in Section~\ref{data_structs}.
\item Section~\ref{alg_overview} provides a high-level overview of the 
new algorithm.
\item Section~\ref{gen_alg_pres} presents the new algorithm.
\item Section~\ref{label_var_section} discusses issues with label variables
which are a key difference between the general algorithm developed in this
section and the Single Homomorphism algorithm developed in the previous section.
\item Details on the SLP operations used in the algorithm and their complexity is covered in Section~\ref{graph_op_section}.
\item Correctness is proven in Section~\ref{rules_sound_section}
and~\ref{failure_section}. 
\item Finally, the complexity proof is covered in Section~\ref{complexity_section}. 
\end{itemize}

\begin{example}\label{string_length_ex}
One consequence of this new graph interpretation is that the label
variable paths, if not compressed, could grow exponentially in length with respect to the
initial set of label variables.  This can be seen using the same
example used to prove the exponential result in Section~\ref{TA_Section}, $\sigma(n)$.
If the algorithm presented below (without compression) is applied to
the system $\sigma(n)$, we do not get an exponential number
of applications of the sum transformation; rather we get label paths
of exponential length.  The growth is due to the path string being
copied and then doubled at each consecutive level.  Although this
doubling of the string leads to the exponential growth, it also
requires the re-use of the string and this suggests the use of string
compression.  Therefore, we keep each of these paths compressed in the
form of straight line programs.  

Consider again the $\sigma(0)$ system and assume there
is a single label variable, $a$, for the initial system. 
If Algorithm~\ref{alg_A_1} is applied but does not use 
string compression the final length of the string
labeling the longest path at level $n$ will be $2^n-1$.
In $\sigma(0)$ (Figure~\ref{graph_exp}) this is 
$2^2 -1 = 3$. For larger $n$ the result is 
undesirably long paths, as seen in Figure~\ref{graph_exp2}. 
However, these strings can easily be compressed via a $SLP$.
 
\begin{figure}[H]
  \centering
\scalebox{.98}{
\input{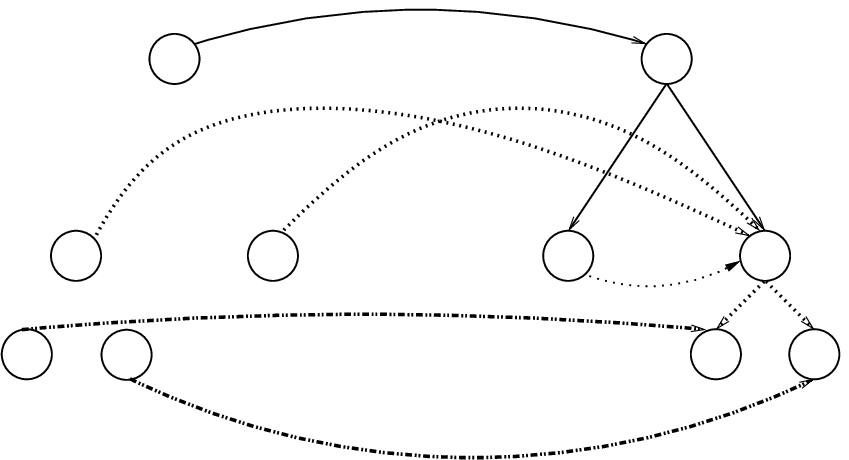}
   }
  \caption{Exponential path length, initial graph}\label{graph_exp}
\end{figure}

\begin{figure}[H]
  \centering
\scalebox{.95}{
   \input{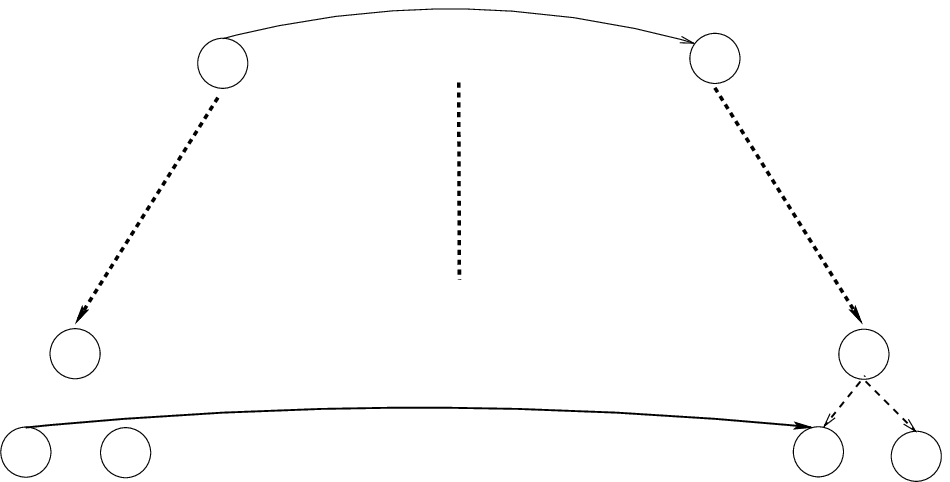}
   }
  \caption{Exponential path length, final graph}\label{graph_exp2}
\end{figure}
\end{example}

\subsection{Data Structures}\label{data_structs}
As in the single homomorphism case we interpret the equations of
a unification problem as graphs.  
\begin{definition}
 we define the following relations:
\begin{itemize}
 \item $X \succ_{r_*} Y$ if $X =  Z \times Y$.
\item $X \succ_{l_*} Z$ if $X =  Z \times Y$.
 \item $X \succ_{m} Y$ if $X \succ_{l_*} Y$ or $X \succ_{r_*} Y$.
 \item $X \succ_{l_+} Y$ if $X = Y + Z$.
 \item $X \succ_{r_+} Z$ if $X = Y + Z$.
\item $X \succ_{a} Z$ if $X = Y + Z$ or $X = Z + Y$.
\end{itemize}
\end{definition}
We denote the transitive closure of a relation $R$ as $R^{+}$.

\begin{definition} \label{LD}
A {\em path labeled dependency graph\/} ($\mathcal{LD}$) is a directed
graph such that the nodes in the graph correspond to variables of the
unification problem $S$.  We form three kinds of edges:

{\em (1) Lateral\/} Edges, where for each equation of the form $X =_{}^? Z \times Y$, 
we have an edge from node $X$ to node $Y$ \underline{labeled} with the top nonterminal of
a $SLP$ generating the {\em label variable\/}, $Z$. 
Label variables are kept as \textit{straight line programs}, 
where the terminals corresponds to the label variable. Each label variable, $Z$,
is given a unique single production $SLP$.
Therefore, lateral edge and path labels correspond to the \underline{top} nonterminal
of the $SLP$ generating the label variables corresponding to those edges.  \\
We denote a path and its label, $\pi$, of one or more lateral edges between nodes $X$ and $Y$ by 
$X \; \stackrel{\pi}{\longrightarrow} \; Y$.
For the general case paths are the composition of any number
of the label variables. A path $\pi$ is notation for a path 
$X_1, \ldots, X_n$ for some
$n \in \mathbb{N}$ and is kept altogether compressed as a SLP.
Therefore, $X \; \stackrel{\pi}{\longrightarrow} \; Y $ 
corresponds to the equations $X =^{?} \pi \times Y$
and is a compact representation of the equation 
\[X =^{?} X_1 \times X_2 \times \ldots \times X_n \times Y\]
where the string generated by $\pi$ is of the form 
$X_1 \cdotp X_2 \cdotp \ldots \cdotp X_n$

{\em (2) Downward\/} Edges, where for each equation of the 
form $X =_{}^? X_1^{} + X_2^{}$, we have
directed edges from node $X$ to node $X_1$ and from node $X$ to node $X_2$.

{\em (3) Relation\/} Edges, where for each node $X$ in the graph such
that there exists a path $X \; \stackrel{\pi}{\longrightarrow} \; Y$ and
for each terminal/label variable $K_i$ in the $SLP$ $\pi$, we have a
\underline{single} edge from $X$ to the {\em node\/} $K_i$ in the
graph.

These edges will only be used for cycle checking and could even be generated
just before the graph is checked for cycles in the algorithm.
\end{definition}
We explain several points below that should help clarify the need for such a graph.
\begin{itemize}
\item The initial $\mathcal{LD}$ graph will be built from an initial
unification problem, $S$, in standard form. That initial graph will not have any composite paths labeled by a $SLP$ 
with more then one production. The composite paths will be added later by the algorithm.
In addition, when constructing the $\mathcal{LD}$ graph each variable $X$ from the
set of label variables is given a unique $SLP$.
For example, a label variable $X$ would be given a $SLP$ $\pi_X \rightarrow X$
and all lateral edges formed by an equation with $X$ as the label variable would
be labeled by $\pi_X$.
This implies that different lateral edges can have the same edge label.
For example, in the $\mathcal{LD}$ graph of Figure~\ref{ld_examp} the edges 
$X \rightarrow Y^{}$ and $L_2 \rightarrow L_3$ have the same $SLP$ label
because they used the same label variable in the equations
$X =_{}^? Z_1 \times Y^{}$ and $L_2 =_{}^? Z_1 \times L_3$.

\item The algorithm presented later will build up the composite paths and 
unlike in the initial
graph it will not in general be the case that for a 
$X \; \stackrel{\pi}{\longrightarrow} \; Y$
all the terminals (label variables) in $\pi$ are $\succ_{l_*}$ related to $X$.
This is the reason for the additional ``Relation'' edges.

\item Lateral Paths are kept as $SLP$s because if not
kept compressed the paths could grow exponential in size, 
during running of the algorithm. By keeping the initial label variables as $SLPs$ 
when we build longer
composite paths we can create the new $SLP$ labels by ``concatenating'' the $SLPs$.

\item 
The {\em Relation\/} edges are only needed during graph cycle checking
operations required by the algorithm. Thus, as the information about
these edges is maintained by the set of terminals for each $SLP$, we
will just generate these edges before cycle checking.
\end{itemize}
\ignore{ 
We denote the length of a path $\pi$ by $| \pi |$, which is simply the
number of variables forming the path ($m$ in Definition~\ref{LD}).}

\begin{example}
Let us consider an example $\mathcal{LD}$ graph for the following system 
of equations ($re$ denotes relation edges).  \[ X =_{}^? X_1^{} + X_2^{}, 
 ~ X =_{}^? \pi_1 \times Y^{}, ~Y =_{}^? \pi_2 \times L^{},
 ~ L =_{}^? L_1^{} + L_2^{}, ~L_1 =_{}^? \pi_5 \times K^{},
~L_2^{} =_{}^{?} \pi_1 \times L_3
\] where the $SLP$s are: $\pi_1 \rightarrow Z_1$, $\pi_2 \rightarrow Z_2$ and
$\pi_5 \rightarrow \pi_3 \pi_4$, $\pi_3 \rightarrow \pi_2 \pi_2$, $\pi_4 
\rightarrow \pi_1 \pi_1$.
The corresponding $\mathcal{LD}$ graph is given in Figure~\ref{ld_examp}.
\end{example}

\begin{figure}
  \centering
   \input{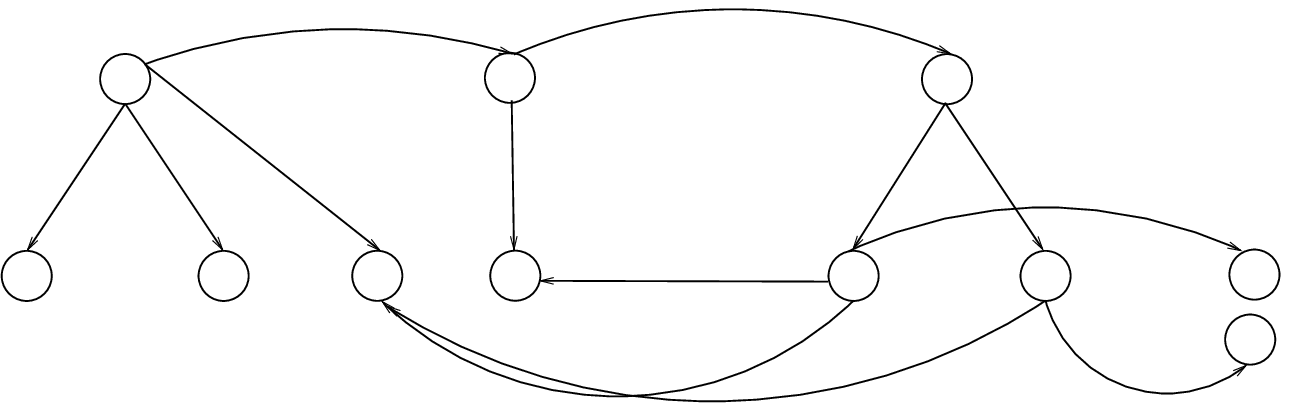}
  \caption{LD graph example}\label{ld_examp}
\end{figure}

\begin{definition} 
The {\em path labeled propagation graph\/}
$\mathcal{LP}$ is a directed simple graph.
Its vertices are the equivalence classes of the symmetric, 
reflexive, and transitive closure of the relation defined by 
$\succ_{r_*}$ on the $\mathcal{LD}$ graph for the same system. 
Edges exist between equivalence classes $\left[ X \right]$ 
and $\left[ Y \right]$ 
if there exist
variables $U \in \left[ X \right]$ and 
$V \in \left[ Y \right]$ such that $U \succ_{a} V$.
\end{definition}

Similar to the typed case we can order the equivalence classes/nodes of $\mathcal{LP}$.
Let~$\sim_r^{}$ stand for the reflexive, symmetric and transitive
closure of~$\succ_{r_*}$. Thus $\sim_r^{}$ defines a set of equivalence
classes over a set of variables. Denote these classes as $\left[ Y \right]_r$.
The $\mathcal{LP}$ graph has exactly these classes as its nodes. 
We can define a strict partial ordering $\gtrdot_r$ on the $\sim_r^{}$-equivalence classes
based on $\succ_a^{}$. That is, $\left[ X \right]_r \gtrdot_r \left[ Y \right]_r$ 
if and only if there exist $ K_1 \in \left[ X \right]_r$ and $K_2 \in \left[ Y \right]_r$
such that $K_1 \succ_a^{} K_2$. This ordering  will again be important 
as it provides an ordering strategy for applying the rules of the unification algorithm. 

Again, we also need  the $\mathcal{LP}$ graph due to a specific
type of non-unifiable system. 
These are systems that require infinite unifiers but will always cause a cycle in the $\mathcal{LP}$ graph. 
An example of this is the following set of equations.
\[ \{
X =_{}^? X_1 + X_2, \;
X =_{}^? V \times X_2 \} \]

\begin{lemma}\label{infinite}
 Let $\mathcal{S}$ be a system of equations with variables
$U, W $ in $S$ such that $U \succ_a W$ and $U \succ_{r_{*}} W$. 
Then $\mathcal{S}$ is not unifiable.
\end{lemma}
\begin{proof}
This system can be seen to cause a cycle not only in $\mathcal{LP}$ but
also in $P(S)$ (see Definition~\ref{orig_PL}). This is due to forming the 
equivalence classes by closure along $\succ_{r}$ related nodes. 
It is shown in~\cite{TidenArnborg87} Lemma 11, if $P(S)$ contains a cycle
there is no unifier for the system $S$ (See Theorem~\ref{TA_Results}).
\end{proof} 

Each time the algorithm updates the $\mathcal{LD}$ graph
it also updates the $\mathcal{LP}$ graph and checks for cycles. Likewise,
if cycles are found the algorithm terminates with failure. 

It can be seen that the propagation graph of~\cite{TidenArnborg87} and the
$\mathcal{LP}$ graph are the same. This is due to fact
that both graphs contain the same equivalence classes and equivalent
edges between the classes. 
Figure~\ref{lp_examp} is an example $\mathcal{LP}$ graph for the same
set of equations used to form the graph of Figure~\ref{ld_examp}.

\begin{figure}
  \centering
   \input{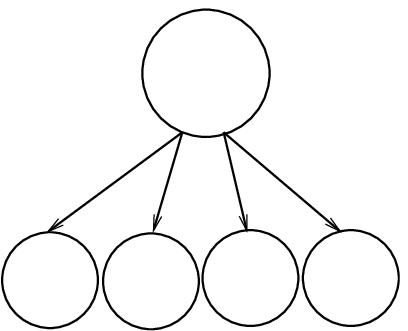}
  \caption{LP graph example}\label{lp_examp}
\end{figure}

The algorithm will work by ``saturating'' the graphs. A set of 
transformation rules is used to either convert the graph into 
a solved form or detect a cycle in the graph. The first case implies 
unifiability and the second non-unifiability.   
During saturation we derive {\em path constraints\/} of the form
$\pi_1^{} =_{}^? \pi_2^{}$ or $\pi_1^{} {\prec}_{}^? \pi_2^{}$. 
The constraint $\pi_1^{} {\prec}_{}^? \pi_2^{}$, is a prefix check 
(i.e., whether the string produced by the $SLP$ $\pi_1$
is a prefix of the string produced by the $SLP$ $\pi_2$)
and $\pi_1^{} =_{}^? \pi_2^{}$, similarly, is an equality check.
In addition to path constraints we will need to
perform several {\em path computations:} specifically we need 
to {\em concatenate paths\/}, {\em compute path suffixes\/} and find
a single pair of {\em mismatched terminals\/} in two equated $SLP$ produced strings, all
without decompressing the $SLPs$.

\subsection{High Level Overview}\label{alg_overview}
Before giving the inference rules and algorithm details let us give a high-level
overview of the process. The algorithm works based on the idea of collecting equations into sets. Each set correspond to the
equations forming one of the equivalence classes. We can then order the equivalence classes in such a way that
if we proceed top down in this order, converting each set into a solved form (processing), we do not have to revisit any class.
This combined with the fact that we don't create new classes requiring processing provides us with a well defined structure
for the execution of the algorithm. Briefly, the algorithm proceeds as follows:
\begin{itemize}
\item  The algorithm begins with a unification problem, $S$, in standard form.
From the set of equations $S$ it generates a graph interpretation, the
$\mathcal{LD}$ graph and makes note of the set of label variables,
$\mathcal{V}$. Note that this process does not discard $S$. 
From the $\mathcal{LD}$ graph the algorithm generate the $\mathcal{LP}$ graph.
The nodes of the $\mathcal{LP}$ graph are the equivalence classes, each of which corresponds to a set of nodes from
the $\mathcal{LP}$ graph. 
\item Next the algorithm reduces the graph to a normal-form by applying the set of cancellation inference rules.
These rules are applied to the entire system which results in a system where no additional cancellation rules can be applied.
All of these rules work directly on the $\mathcal{LD}$ graph, except rule $(0)$ which works on the set~$S$.
\item Next the algorithm selects an equivalence class to process based on the class ordering. It then applies a set of
rules on the nodes of the $\mathcal{LP}$ graph which correspond to that equivalence class. The effect of these rules is
to transform the equivalence class into a compressed solved form. 
\item This process is repeated along with error checking, each time reducing the number of classes remaining to be processed. 
\item Finally, it could happen that during the process two label variables are equated. If this occurs the algorithm updates $S$
with the new equality (rule $(0)$), rebuilds the two graphs and the process of processing equivalence classes is restarted.
Since there is only a finite number of label variables, which cannot be increased, each time two label variables are equated 
the number of label variables is reduced. Therefore the total number of times the process can be restarted is equal to the initial number of label variables in $S$.    
\end{itemize}

\subsection{Algorithm Presentation}\label{gen_alg_pres}
We first present the set of inference rules (Fig~\ref{Gen_Inf_Rules}) for a unification problem $S$ 
in standard form. The rules are applied to the
graph $\mathcal{LD}$, except rule $(0)$ which is applied to $S$, and as that graph is updated the $\mathcal{LP}$ graph
is updated.  

\begin{figure}[ht]
  \centering
\scalebox{.85}{
\fbox{
\begin{tabular}{lc}
$(0)$ &  $\vcenter{
\infer[\qquad \mathrm{if} ~ U ~ \mathrm{occurs ~ in} ~ S ]{ \{ U \mapsto V \}(S) \cup 
\, \{U =_{}^? V\} }
      { S ~ \uplus ~ \{U =_{}^? V\} }
}
$\\[+25pt]
$(i)$ &  $\vcenter{
\infer{U=_{}^? U_1 + U_2 , \; U_3 = U_1 , \; U_4 = U_2}
{ U=_{}^? U_1 + U_2 , \; U =_{}^? U_3 + U_4 \; } }
$\\[+25pt]
$(ii)$ &  $\vcenter{
\infer{X \, {=}_{}^? \, Z, \; \qquad  X \; \stackrel{\pi}{\longrightarrow} \; Y}
{X \; \stackrel{\pi}{\longrightarrow} \; Y, \; \qquad Z \; \stackrel{\eta}{\longrightarrow} \; Y \;  \qquad \eta \, = \, \pi}
}
$\\[+25pt]
$(iii)$ &  $\vcenter{
\infer{Y \, {=}_{}^? \, Z, \; \qquad X \; \stackrel{\pi}{\longrightarrow} \; Z}
{X \; \stackrel{\eta}{\longrightarrow} \; Y, \; \qquad X \; \stackrel{\pi}{\longrightarrow} \; Z, \; \qquad \eta \, = \, \pi}
}
$\\[+25pt]
$(iv)$ &  $\vcenter{
\infer{Y \; \stackrel{{\eta}^{-1} \, \pi}{\longrightarrow} \; Z, \; \qquad X \; \stackrel{\pi}{\longrightarrow} \; Z}
{X \; \stackrel{\eta}{\longrightarrow} \; Y, \; \qquad X \; \stackrel{\pi}{\longrightarrow} \; Z, \; \qquad \eta \, \prec \, \pi}
}
$\\[+25pt]
$(v)$ &  $\vcenter{
\infer{\eta \, {=}_{}^? \, \pi }
{X \; \stackrel{\eta}{\longrightarrow} \; Y, \; \qquad X \; \stackrel{\pi}{\longrightarrow} \; Y, \; \qquad \eta \, \neq \, \pi}
}
$\\[+25pt]
$(vi)$ &  $\qquad \qquad$ $\vcenter{
\infer[\qquad \mathrm{if} ~ | \eta | \, < \, | \pi |]{\eta \, {\prec}_{}^? \, \pi}
{X \; \stackrel{\eta}{\longrightarrow} \; Y, \; \qquad X \; \stackrel{\pi}{\longrightarrow} \; Z, \; \qquad \eta \, \not\prec \, \pi}
}
$\\[+25pt]
$(vii)$ &  $\vcenter{
\infer{W =_{}^? W_1^{} + W_2^{}, \; U_1 \; \stackrel{\eta}{\longrightarrow} \; W_1^{},
 \; U_2 \; \stackrel{\eta}{\longrightarrow} \; W_2^{}, \; U \; \stackrel{\eta}{\longrightarrow} \; \; W }
{ U \; \stackrel{\eta}{\longrightarrow} \; \; W, \; U =_{}^? U_1 + U_2, \; W =_{}^? W_1^{} + W_2^{} \; } }
$\\[+25pt] 
$(viii)$ &  $\vcenter{
\infer{W =_{}^? W_1^{} + W_2^{}, \; U_1 \; \stackrel{\eta}{\longrightarrow} \; W_1^{},
 \; U_2 \; \stackrel{\eta}{\longrightarrow} \; W_2^{}, \; U \; \stackrel{\eta}{\longrightarrow} \; \; W }
{ U \; \stackrel{\eta}{\longrightarrow} \; \; W, \; U =_{}^? U_1 + U_2, \; } }
$\\[+25pt] 
$(ix)$ &  $\vcenter{
\infer[\qquad \mathrm{if} ~ \mathcal{LP} ~\mathrm{or} ~ \mathcal{LD} ~\mathrm{are ~cyclic} ]{FAIL}
{S}
}
$\\[+25pt]
$(x)$ &  $\vcenter{
\infer{X \; \stackrel{\tau = (\pi \eta)}{\longrightarrow} \; Z, \; \qquad Y \; \stackrel{\eta}{\longrightarrow} \; Z}
{X \; \stackrel{\pi}{\longrightarrow} \; Y, \; Y \stackrel{\eta}{\longrightarrow} \; Z }
}
$\\
\end{tabular} } }
 \caption{Inference Rules for the One-sided Distributivity Decision Procedure.}
  \label{Gen_Inf_Rules}
\end{figure}

\subsubsection*{Rule Zero}
Rule $(0)$ is simply a variable replacement rule but it has a special
action on label variables: 
{\em if a label variable is equated to a
non-label variable, then the non-label variable is replaced by the
label variable.\/} This rule acts directly on the system~$S$
by doing variable replacement whenever there is an equation
between two variables. 
That is, for an equation of the form $U = V$
between two variables, the rule replaces all occurrences of $V$ in
$S$ with $U$. If one of the variables is a label variable, say $V$,
and one is not, say $U$, then the non-label variable is replaced in
$S$ by the label variable. So in this example, all occurrences of $U$ are replaced by $V$. Therefore, after any variables are equated, we
apply this rule {\em eagerly.\/} Note, \emph{whenever a rule
creates an equation of the form $X=^{?}Y$, those two nodes in the graph are equated and rule $(0)$ applies that equation to the set $S$}.
Therefore, rule $(0)$ is the only rule that acts on and changes the set
of equations $S$. All other rules  modify the two graph data structures.

Rule $(i)$ is due to the cancellative nature
of the $+$ operator and directly corresponds to the canceling
operation in~\cite{TidenArnborg87} (see Figure~\ref{TA_Inf_Rules}).
Rules $(ii)$, $(iii)$ and $(iv)$ are
due to the cancellative nature of $\times$~(\cite{TidenArnborg87}).
Rules $(v)$ and $(vi)$ check the path constraints and attempt to find
label variables that have to be equated in order to satisfy the path
constraint.  These rules and rules $(iv)$ 
and $(x)$ are explained in more
detail in Section~\ref{graph_op_section}.  Rules $(vii)$ and $(viii)$
directly correspond to the splitting rule 
(Rule (d) of Figure~\ref{TA_Inf_Rules}) of~\cite{TidenArnborg87} and
are direct consequences of the distributive axiom. These two versions
are just modifications to work in the modified graph setting. The
difference between the two rules is that rule $(viii)$ creates new variables ($W_1$ and $W_2$) and rule 
$(vii)$ does not.
Rule $(ix)$ is a failure rule, which corresponds to detecting a cycle
in the graphs in the Tid\'{e}n-Arnborg algorithm. 
Finally, rule $(x)$ is a path
completion rule, justified by the soundness of variable
replacement. This rule is also responsible for building the $SLPs$
with more then one production. The rule creates a new $SLP$,
$\tau$, corresponding to the ``concatenation'' of the two $SLP$s $\pi$
and $\eta$. More details on rule~$(x)$ are given in
Section~\ref{graph_op_section}.  
  
\ignore{
\begin{figure}[ht]
  \centering
\scalebox{.90}{
\fbox{
\begin{tabular}{lcc}
(d) & $\qquad$ & $\vcenter{
\infer{X \; \stackrel{l + m}{\longrightarrow} \; Z}
{X \; \stackrel{l}{\longrightarrow} \; Y, \; \qquad Y \; \stackrel{m}{\longrightarrow} \; Z}
}
$\\[+20pt]
(e) & & $\vcenter{
\infer{ U_1 \; \stackrel{l}{\longrightarrow} \; W_1^{}, \; \; U_2 \; \stackrel{l}{\longrightarrow} \; W_2^{} }
{ U \; \stackrel{l}{\longrightarrow} \; \; W, \; U =_{}^? U_1 + U_2, \; \; W =_{}^? W_1^{} + W_2^{} }
}
$\\[+20pt]
(f) & $\qquad$ & $\vcenter{
\infer{Y \; \stackrel{m - l}{\longrightarrow} \; Z}
{X \; \stackrel{l}{\longrightarrow} \; Y, \; \qquad X \; \stackrel{m}{\longrightarrow} \; Z, \; \qquad l < m}
}
$
\end{tabular}
} }
 \caption{Inference Rules for Tracking Path Lengths.}
  \label{Path_Length_Inf_Rules}
 \end{figure} 
}

We also keep and update the length of the string each $SLP$ generates. 
Note, that this information can be efficiently  computed in a bottom up manner for 
any $SLP$ since productions ending in a terminal symbol have string length $1$ and
productions with two non-terminals have length equal to the sum of the lengths of the strings
generated by the two non-terminals. However, since we build our $SLP$ bottom up we can
keep track of this information using simple addition and subtraction when 
constructing new $SLPs$ through concatenating and taking the suffix.  							 
 
\begin{algorithm}[ht]
\caption{One-sided Distributive Unification}
\label{alg_A_1}
\begin{algorithmic}
\vspace{0.1in}
\STATE (Input: A system of equations in standard form)
\STATE \textbf{(1: Generate data structures)}
Generate the 2 graphs, $\mathcal{LD}$ and $\mathcal{LP}$. 
Make a note of the initial label variables in
$S$; denote this set as $\mathcal{V}$.\\[+2pt]

\STATE \textbf{(2: Clean up the system)}
Exhaustively apply the following composite rule:
\[ (0 + i + ii + iii) \]

\STATE \textbf{(3: Error checking)} 
Apply graph cycle checking to the graphs (i.e., rule $(ix)$). If a cycle
is found {\em stop\/} with failure. If the graphs have no cycles
and are in dag-solved form, {\em exit\/} with success. \\[+2pt]

\STATE \textbf{(4: Process equivalence class)} 
Select an equivalence class based on the strict
partial ordering $\gtrdot_r$.
That is, we select the largest element of $\gtrdot_r$ that
has not yet been processed. Thus, if we select the class
$\left[ X \right]_r$ then there does not exist a class $\left[ Y \right]_r$
such that $\left[ Y \right]_r$ has not been processed and 
$\left[ Y \right]_r$~$\gtrdot_r^{+}$~$\left[ X \right]_r$\footnotemark.\\[+2pt]

Process the selected class using the following composite rule:
\[(v + vi)^{}(iv)^{!}(x)^{!}(viii)(vii)^{!}\]
\emph{Rule $(x)$ is applied by starting with the
sink node of the path and working back to the start node of the
path}. 
\emph{Rule $(iv)$ is applied based the $\succ_{r_*}$ partial ordering, starting from 
the source nodes and working down to the sinks.}
In addition, if rule $(v)$ or rule $(vi)$ is applied 
label variables will be equated. 

\STATE \textbf{(5: Checking)}
If any of the variables in $\mathcal{V}$ are equated 
go back to Step~1 else go back to Step~2. 
\vspace{0.1in}
\end{algorithmic}
\end{algorithm}

\footnotetext{Again, if $\gtrdot_r$ is not
strict partial ordering there must be a cycle in the $\mathcal{LP}$ graph}

Algorithm~\ref{alg_A_1} uses the following notation. 
Let $r_1$ and $r_2$ denote inference rules.  
Then, $r_1^{!}$ indicates {\em exhaustive\/} application of the rule $r_1$.
The composite rule $r_1^{!} r_2$ means, apply $r_1$
until it cannot be applied any more and then try to apply $r_2$.
Note that even if $r_1$ cannot be applied the rule 
$r_1^{!} r_2$ can still be used if $r_2$ can be applied.
Thus $r_1^{!}$ does not indicate that $r_1$ \textit{must} be applied
but rather that if $r_1$ can be applied we do so exhaustively. 
$r_1 + r_2$ indicates choice: apply rule $r_1$ \textit{or} rule $r_2$.
Therefore, the last composite rule in Algorithm~\ref{alg_A_1} 
implies that rule $(vii)$ has the 
lowest priority and that rule $(viii)$ is only applied \emph{once} in
the processing of a single equivalence class.

Algorithm~\ref{alg_A_1} is presented above. Let us now give some additional 
explanation of what each step accomplishes before proceeding to the proof details.
\begin{enumerate}
\item The first generates the two data structures
needed to check error conditions. 
\item Step two consists of an exhaustive application of the ``cancellation'' inference rules, 
i.e., rules $(0), ~(i), ~(ii), ~(iii)$. These rules are the simplest rules as they
either reduce the number of variables, rule $(0)$, or reduce the number of edges in the graph,
rules $~(i), ~(ii), ~(iii)$. In addition, the rules don't create any new SLPs or edges. By applying these inference rules first we reduce the problem to a ``normal-form'', where change
can now only occur via rules $(iv) - (x)$.  
Some of the rules must check equality between 
SLP, which can be done in polynomial time, see section~\ref{slp_ops}.  
Lemma~\ref{sound} shows that the inference rules
are sound and Lemma~\ref{step2_bound} shows that this step of the algorithm is polynomially bounded.
\begin{enumerate}
\item Note that redundant edges are removed by rule $(iii)$. 
That is, if there are two edges
\[
X \; \stackrel{\pi}{\longrightarrow} \; Y, \; \qquad X \; \stackrel{\eta}{\longrightarrow} \; Y \; \text{ and } \qquad \eta \, = \, \pi
\] 
Then, rule $(iii)$ will remove one of the redundant edges.  
\end{enumerate}
\item Step three is an error checking step which corresponds to the cycles checking of the two graphs. The
correctness of the step is due to the fact that the system is simple and therefore cycle, or occur-checks,
are errors. See section~\ref{failure_section} for full proofs. 
\item Step four is the step responsible for processing equivalence classes and corresponds to the application
of the inferences rules $(v), ~(vi), ~(iv), ~(x), ~(viii)$ and $(vii)$ in a specific order. The order that rules are applied is 
important as it ensures that the number of steps required to process any equivalence class is bounded by a polynomial, 
see Lemma~\ref{equiv_proc}.  
\item Step five checks to see if any label variables were equated and if so, it goes back to step one, where we rebuild the
graphs. This is done to ensure we don't miss any variables which are part of the compressed labels. Since the number
of label variables is reduced at each application, Lemma~\ref{label_vars_lemma_2} shows the number of times this can happen
is bounded by the initial number of label variables.    
\end{enumerate}
\ignore{
The remainder of this section has the following structure.  In
Section~\ref{label_var_section} issues with the label variables are
clarified. In Section~\ref{graph_op_section} we explain in more detail
the issues and methods for using $SLPs$ and for doing cycle checking. Sections~\ref{rules_sound_section} and~\ref{failure_section} cover the correctness proofs.  
The complexity of the Algorithm~\ref{alg_A_1} is analyzed in Section~\ref{complexity_section}. 
}

\subsection{Label Variables}\label{label_var_section}
We need several results about the label variables and their 
interaction with the new variables. 
Let $\mathcal{V}_0$ denote the set of initial label variables for a system $S$ and
$\mathcal{V}$ the set of label variables at any point during the application of
Algorithm~\ref{alg_A_1}.
Let $\mathcal{Z}$ denote the set of fresh variables created by rule~$(viii)$ 
during 
application of Algorithm~\ref{alg_A_1}.
\begin{lemma}\label{label_vars_lemma_1}
During and after the application of Algorithm~\ref{alg_A_1}, $\mathcal{V} \, \cap \, \mathcal{Z} = \emptyset$.
\end{lemma}
\begin{proof}
By the definition, it is not possible to apply rule $(0)$ such that a newly created variable 
is made a label variable. 
The only way to make a new variable a label variable is to create a new
lateral edge and make its label a new variable. The rules creating new lateral
edges are $(iv)$, $(vii)$ and $(viii)$. But, the labels of these edges are all 
composed of pre-existing label variables. 
\end{proof}

\begin{lemma}\label{label_vars_lemma_2}
During and after the application of Algorithm~\ref{alg_A_1} on a system $S$ in standard form, 
$|\mathcal{V}| \leq |\mathcal{V}_0|$
\end{lemma}
\begin{proof}
Follows directly from rule $(0)$ and Lemma~\ref{label_vars_lemma_1}. 
\end{proof}

Actually we can get a similar result for the original Tid\'{e}n-Arnborg
algorithm if we also assume that variable replacement in that
algorithm replaces newly created variables by original variables. Stated
more precisely we get the following lemma.

\begin{lemma}
Let $\mathcal{V}_l^{}$ denote the set of left multiplication variables, i.e., 
$Z \in \mathcal{V}_l^{}$ iff there exists an equation of the 
form $X =Z \times Y$ for some
variables $X$ and $Y$.
Let {\em Sum\/} denote the sum transformation operation 
as defined in~\cite{TidenArnborg87} 
(See Section~\ref{TA_Section} to recall the definition).
Also, assume that if a pre-existing variable is equated to 
a variable created by {\em Sum,\/}
the new variable is replaced by the pre-existing one. 
Then it is never the case that a new variable
created by {\em Sum\/} is in $\mathcal{V}_l^{}$.
\end{lemma}
\begin{proof} 
The $Sum$ operation does not create new
left multiplication variables. Therefore, the only way to get a
new variable $Z$ into $\mathcal{V}_l$ is by equating it with a
variable already in $\mathcal{V}_l$.
\end{proof}

Thus, \emph{we can assume that there will never be an equation of the form
$X =^{?} Z \times Y$, where $Z$ is a fresh variable created by Algorithm~\ref{alg_A_1}}.  
\ignore{
However, we could have a
relation $Z \succ_{l_*} X$ where $X$ is some pre-existing
variable. This implies that there exists a variable $K$ and an
equation of the form $K =^{?} L \times Z$.  This leads to the observation
(assuming the modified variable replacement method) that we only need
to apply the sum transformation like rule creating new variable (rule~$(viii)$) 
to the last variable in a path of $\succ_{r_*}$-related variables and we need only do this
if the last node in the path has no downward edges. 
\emph{This is essentially the reason we don't need to generate all the new variables in the same way as is done in~\cite{TidenArnborg87}.}
} 

\subsection{Graph and SLP Operations}\label{graph_op_section}
We first examine the problem of graph cycle checking and then we cover
the details of the $SLP$ operations.

\subsubsection{Graph Cycle Checking and Updating}
The $\mathcal{LD}$ graph is updated by the algorithm as the inference rules
operate on it. The $\mathcal{LP}$ graph is built from the $\mathcal{LD}$ graph and thus
can be updated after updating the $\mathcal{LD}$ graph. 
We note that the $\mathcal{LD}$ and $\mathcal{LP}$ graphs can use 
standard cycle checking 
algorithms. With the additional observation that 
we can add the relation edges
in 
polynomial time with respect to the
number of nodes in the $\mathcal{LD}$ graph, we get the
following Lemma.

\begin{lemma}\label{graph_cycle_poly}
The $\mathcal{LP}$ and $\mathcal{LD}$ graphs for a system $S$ in standard form
can be checked for cycles in polynomial time with respect to 
the size of $S$.
\end{lemma}
\ignore{
Because checking $X \; \succ_d \; X$ is essentially cycle checking in
the $\mathcal{LD}$ graph Lemma~\ref{graph_cycle_poly} also implies
that $X \; \succ_d \; X$ can be checked in polynomial time.
}
\subsubsection{SLP Operations}\label{slp_ops}
Algorithm~\ref{alg_A_1} requires the use of some type of string
compression due to the need to keep path labels from growing
exponentially. But, we still need to know how the label variables are
related along paths, e.g., for error checking. Therefore, we cannot
just keep a {\em set\/} of the variables forming the path, the
terminals of the $SLP$, because this removes essential information.
We first examine how the $SLP$ are formed and then we will discuss
how the operations are used by each rule in Algorithm~\ref{alg_A_1} 
along with a presentation of their complexity. For convenience, 
Table~\ref{table0} gives a listing of the $SLP$ algorithms
required by Algorithm~\ref{alg_A_1} and a listing of where 
polynomial time algorithms have been developed and studied for that problem. See~\cite{Lohrey12, Rytter99} for surveys on algorithms on compressed strings, including
$SLPs$.
More details are given in the following discussion 
\subsubsection*{Forming SLPs}
We first encode the label variables as $SLPs$. Each unique label variable is encoded as a unique $SLP$. 
For example, when creating the $\mathcal{LD}$ graph for
two equations $X=^{?} Y * Z$ and $K =^{?} Y * L$ only
\underline{one} $SLP$ is created, $\pi_Y \rightarrow Y$,
and two edges are labeled by that $SLP$, i.e., by the top nonterminal
$\pi_Y$. 
Then, larger or additional $SLP$s are formed, bottom up, by the
inference rules $(x)$ and $(iv)$. 
In addition, \textit{we only keep a single copy of each unique SLP}.
This implies we only keep the \underline{set} of all productions. When creating a new larger $SLP$ we need only create a new top 
production.
For example, if we have two pre-existing $SLPs$ $\pi_i$ and $\pi_j$
that we wish to concatenate we don't need to duplicate all the productions;
simply adding a top production $\pi_{(ij)} \rightarrow \pi_i \pi_j$
to the set of productions is sufficient. Likewise, when constructing a suffix,
we may need to create new productions that are added to the set of productions
but we do not delete the productions contained in the prefix since they generate other
$SLPs$. Note that rules $(vii)$ and $(viii)$ don't
create new $SLPs$ but just use pre-existing ones.

\begin{table}[h]
\begin{center}
    \begin{tabular}{|p{8.0cm}|l|}
    \hline
\emph{$SLP$ operations required by Algorithm~\ref{alg_A_1}}& 
\emph{Reference to Polynomial Algorithm}   \\ \hline
The concatenation of two $SLP$ & ~\cite{Claude09, levy04, Levy08}\\ \hline
$SLP$ equality & ~\cite{Lifshits07, MiyazakiST97, Rytter03, plandowski94} \\ \hline
$SLP$ prefix and suffix & ~\cite{Lifshits07, gascon09, gascon11, Levy08} \\ \hline
Find one pair of non-equal terminals in a pair of non-equal $SLPs$
& ~\cite{gascon09, Lifshits07, MiyazakiST97}  \\ \hline
    \end{tabular}
\caption{Algorithms for $SLP$ Operation}
  \label{table0}
\end{center}
\end{table}  
\noindent
\subsubsection*{Using SLPs} We now examine the $SLP$ operations used by
each rule and their complexity.

\noindent 
\underline{Rule $(x)$:} Rule $(x)$ forms a new $SLP$ by 
concatenating two existing $SLPs$.
\emph{Rule $(x)$ is applied by starting with the
sink node of the path and working back to the start node.}\\
 This ensures a minimal number of applications of rule $(x)$.
To concatenate two $SLPs$, $I = (\Sigma, N_I, P_I)$ and $J = (\Sigma, N_J, P_J)$,\footnote{
For every $SLP$ the set of terminals will be a subset of $\Sigma$, the initial set of label
variables.}
we create a new $SLP$, $K =(\Sigma, N_K, P_K)$.
Let $\pi_I$ and $\pi_J$ be the top nonterminals of $I$ and $J$ respectively.
Then, $N_K = N_I \cup N_J \cup \left\lbrace \pi_K \right\rbrace $ and $P_K = P_I \cup P_J \cup 
\left\lbrace \pi_K \rightarrow \pi_I \pi_J \right\rbrace $.
This is a simplified version, with just two $SLPs$, of the method presented
in~\cite{Levy08}, for concatenating $n$ strings. There it is shown,
by a constructive proof,\footnote{
We have replaced ``singleton context free grammar''
with $SLP$ in the statement, just to stay consistent with the naming in this paper.}
that the $SLP$, $G$, generating the new string satisfies
$|G^{'} | \leq | G | + n -1 $ and 
$ depth(G^{'}) \leq depth(G) + \lceil log(n) \rceil $.
Rule $(x)$ is just concatenating the $SLPs$ but we could also balance the resulting 
$SLP$. It is shown in~\cite{Rytter03} that for a $SLP$, $G$, generating a text of length $m$
with $n$ rules we can construct a $SLP$, $G'$, in $\mathcal{O}(n~log(m))$ time, such that
$G'$ has a depth of $\mathcal{O}(log(m))$ and $\mathcal{O}(n~log(m))$ rules. This could
improve results which depend on the depth of a $SLP$. However, for our purposes we will
use the simple concatenation method as it is sufficient for our results and allows for a simpler
complexity analysis. The following result easily follows.
\begin{lemma}\label{slp_concat_bound}
Let  $I = (\Sigma, N_I, P_I)$ and $J = (\Sigma, N_J, P_J)$ be two $SLPs$. 
Then we can construct in linear time, without decompression, 
a $SLP$ $K =(\Sigma, N_K, P_K)$ that generates the concatenation of the two strings generated by $I$ and $J$ such that  $| K | = |P_I \cup P_J| + 1 $ and 
$depth(K) \leq max \left\lbrace  depth(I), depth(J) \right\rbrace  +  1 $.
\end{lemma}
\noindent
Additional algorithms and notes on concatenation can be found 
in~\cite{Claude09, levy04, Rytter03}.

\noindent
\underline{Rules $(ii)$ and $(iii)$:} These two rules require that we can 
decide if two compressed strings are equal, $\pi_1 =^? \pi_2$. The area
of fully compressed pattern matching is an active area and there are
many algorithms that will solve this problem in polynomial time
($\mathcal{O}(n^{3})$ time for a $SLP$ of size $n$~\cite{Lifshits07}). 
We cite the following, non-exhaustive, list of papers
for excellent algorithms;~\cite{Lifshits07, MiyazakiST97, Rytter03, plandowski94}.

\noindent
\underline{Rule $(iv)$:} 
We can partially order the nodes in each equivalence class based on the 
lateral edges, i.e., based on the $\succ_{r_*}$ relation. 
\emph{Rule $(iv)$ is applied based on this partial ordering, starting from 
the source nodes and working down to the sinks.}\\ 
We do not
apply rule $(iv)$ to a node $X$ if rule $(iv)$ can be applied to a node $Y$
such that there is a lateral path from $Y$ to $X$.  
Rule $(iv)$ requires that we can decide if one
$SLP$ $\pi_1$ is a prefix of an $SLP$ $\pi_2$, $\pi_1 \prec_{}^? \pi_2
$, in polynomial time with respect to $\pi_2$. This problem has been
solved in~\cite{Lifshits07}, $\mathcal{O}(n^{3})$ time for a $SLP$, $\pi_2$, of size $n$. 
We also need to extract the suffix in
compressed form, $\pi_3 = {\pi_1}_{}^{-1} \pi_2$. Because we build the
$SLP$s bottom up and keep the length information. A simple
polynomial-time recursive algorithm can accomplish this. 
See also~\cite{gascon09, Levy08, Lifshits07, schmidtschauLIPIcs2012} 
for additional efficient methods
for computing the suffix (and prefix). 
For example, it has been shown in~\cite{Levy08}
that if $G$ is a $SLP$ generating the word $v$, then for any suffix $v^{'}$
there exists a $SLP$ $G^{'}$ that generates $v^{'}$ and satisfies
$|G^{'}| \leq |G| + depth(G)$ and $depth(G^{'}) \leq depth(G)$.
For completeness a simple, 
but polynomial, suffix algorithm is presented in Appendix A 
and from this algorithm we have the following result. 
 
\begin{lemma}\label{slp_suffix_bound}
Let  $I = (\Sigma, N_I, P_I)$ and $J = (\Sigma, N_J, P_J)$ be two $SLPs$
such that the string generated by $J$ is a prefix of the string 
generated by $I$. Then, in $\mathcal{O}(|I|^{4})$ time a $SLP$ 
$K =(\Sigma, N_K, P_K)$ can be constructed that generates the suffix of $I$ after removing the prefix $J$ such that $| K | \leq |I| + depth(I) $ and $depth(K) \leq depth(I)$.  
\end{lemma}

\noindent
\underline{Rules $(v)$ and $(vi)$:} These rules handle the
situation where two label paths should be equal, or one a prefix of
the other, but are found not to be.  We then need to check if they can
be made equal. We accomplish this by finding at least one pair,
$(X,Y)$, of terminals (label variables) in the corresponding $SLP$s
such that these terminals form a mismatch, $X \neq Y$. One pair will do
for each application of the rule because by the cancellative nature of
$\times$, all mismatched pairs of terminals \textit{must} be equated.
Therefore, we do not have to try all different combinations
of setting pairs equal or unequal.  It is sufficient to select the first mismatch, equate
the variables and construct the resulting new problem. 
Note that we are finding a single pair of terminals that form a
mismatch in the string, not finding {\em all\/} the positions where
the strings generated by the $SLPs$ differ, a NP-hard
problem~(\cite{Lifshits07}). In \cite{gascon09} the authors have developed
a nice polynomial, $\mathcal{O}(n^3)$, algorithm for finding the first mismatch.  
A mismatch can also be found using the
algorithms in~\cite{Lifshits07, MiyazakiST97} or by a simple recursive
algorithm, using the $SLP$ equality algorithm
of~\cite{Lifshits07} as a subroutine. The result is a simple
$\mathcal{O}(n_{}^{4})$ algorithm, $n$ being the size of the largest $SLP$.

The way rules $(v)$ and $(vi)$ work in Algorithm~\ref{alg_A_1} is
if in the $\mathcal{LD}$ graph one of the rules is satisfied, then 
a pair of label variables will be found (by the $SLP$ algorithm)
and equated (through the use of rule~$(0)$).
This will cause the set of label variables, $\mathcal{V}$, to be reduce and 
thus the number of label variables in the system $S$ to be reduced.
The algorithm then returns to step~1 
and rebuilds a new $\mathcal{LD}$ graph from the newly modified system.

\subsection{Correctness}\label{rules_sound_section}
We now examine the correctness of the above algorithm. 
Rather then reconstructing all the proofs from ``scratch'' 
we can reuse some result proven by Tid\'{e}n and Arnborg
since we are working on a compressed version of the original algorithm. 

\begin{lemma}\label{sound}
The non-failure rules of Algorithm~\ref{alg_A_1} maintain the set of unifiers.
\end{lemma}
\begin{proof}
The new rules are essentially equivalent to the original set of rules, only modified
to work on a compressed version of the problem.  
This can be seen by considering a path 
\[ X \; \stackrel{\pi}{\longrightarrow} \; Z\] 
and remove the compression. 
The path is a graphical representation of the equation.
\[X =^{?} \pi \times Z\]
where the string generated by $\pi$ is of the form 
$X_1 \cdot X_2 \cdot \ldots \cdot X_n$, for some 
label variables $X_1, \ldots, X_n$.
This is a compressed form of the following equation.
\[X =^{?} X_1 \times X_2 \times \ldots \times X_n \times Z\]
This was constructed from $n$ equations in standard form using variable replacement.
These equations are of the form
\[X =^{?} X_1 \times K_1, ~K_1 =^{?} X_2 \times K_2,  ~\ldots, ~K_{n-1} =^{?} X_n \times Z\]

Therefore, the result follows from 
Theorem~\ref{TA_sound} which is proven in~\cite{TidenArnborg87}. Let us now examine the rules.
\begin{itemize}
\item Rule $(x)$ and rule $(0)$ follow from variable replacement.
\item Rule $(i)$ follows from Theorem~\ref{TA_sound} part~1.

\item Now consider rules $(ii)$ 
and $(iii)$, since $\pi = \eta$, these rules follow from
$|\pi|$ applications of Theorem~\ref{TA_sound} part~1. 
The same holds for rule $(iv)$ except that $\eta$ is a prefix.
\item Rules $(v)$ and $(vi)$ 
also follow from Theorem~\ref{TA_sound} part~1 because by
part~1 all the label variables in $\pi$ would have 
to be equated to the corresponding 
variables in $\eta$, so we are safe in selecting one pair.
More specifically, consider two paths and $(v)$.   
\[ X \; \stackrel{\pi}{\longrightarrow} \; Z, \; X \; \stackrel{\eta}{\longrightarrow} \; Z \]
correspond to two equations, uncompressed
\[X =^{?} X_{i_1} \times X_{i_2} \times \ldots \times X_{i_n} \times Z, \; 
~X =^{?} X_{j_1} \times X_{j_2} \times \ldots \times X_{j_n} \times Z\]
Both could be expanded out into standard form and by 
applying Theorem~\ref{TA_sound} part~1
we equate each pair $X_{i_l} = X_{j_l}$.
\item Finally, rules $(viii)$ and $(vii)$ follow from
Theorem~\ref{TA_sound} part~2. \qedhere
\end{itemize}
\end{proof}

Note that rule $(ix)$ is a failure condition that is handled by cycle checking the graph. 

\begin{lemma}\label{complete_sub1}
If Algorithm~\ref{alg_A_1} exits with success on a system
$S$ in standard form, then $S$ is unifiable.
\end{lemma}
\begin{proof}
The result follows from Lemma~\ref{sound} and 
the fact that the set of inference rules transforms
$S$ into a dag-solved form, which implies unifiability~\cite{Gallier89}. 
This can be seen by examining the set of rules and the definition of dag solved-form.\\
Part (a) is satisfied because if there existed some $X_i$ such that
$X_i =^{?} t_{1}$ and $X_i =^{?} t_{2}$ (for terms $t_1$ and $t_2$) one of the inference rules
would be applicable.\\
Part (b) is satisfied because there are no cycles in the graph and thus the
equations can be arranged in the proper order. 
\end{proof}

\begin{lemma}\label{complete_sub2}
If Algorithm~\ref{alg_A_1} terminates with failure on a system $S$ in standard form, 
then $S$ is not unifiable.
\end{lemma}
\begin{proof}
Follows from Lemma~\ref{failure_lemma}.
\end{proof}
From these results we get the following.
\begin{theorem}
The decision Algorithm~\ref{alg_A_1} is correct.
\end{theorem}
\begin{proof}
Follows from Lemma~\ref{complete_sub1} and Lemma~\ref{complete_sub2}.
\end{proof}

\subsection{Failure Conditions}\label{failure_section}
Graph cycle checking is employed to detect failure conditions. We argue in this
section that if a cycle is found this corresponds
to a non-unifiable system. 

\begin{lemma}\label{failure_lemma}
A system $S$ in standard form is not unifiable if there exists a cycle
in any of the corresponding $\mathcal{LP}$ or $\mathcal{LD}$ graphs for that system.
\end{lemma}
\begin{proof}
We consider the following cases:\\
\underline{Case 1:} Assume that
the $\mathcal{LD}$ graph for a system $\mathcal{S}$ contains a
cycle. Then the cycle was created by zero or more applications of the
inference rules and implies that a variable is a proper subterm of
itself. By Theorem~\ref{TA_cycle} these cycles correspond to 
non-unifiable systems.

\noindent
\underline{Case 2:} Assume the the $\mathcal{LP}$ graph for a system
$\mathcal{S}$ contains a cycle. This implies there is a cycle between
the $\sim_r$ equivalence classes.  It is shown in~\cite{TidenArnborg87} 
that cycles between  $\sim_r$ equivalence classes of this form correspond 
to non-unifiable systems due to the need for an infinite unifier 
(see Theorem~\ref{TA_Results}).
\end{proof}

Therefore, if cycles are found we can safely conclude that the system
is not unifiable and return an error.

Finally, one could ask if some infinite systems (of Lemma~\ref{infinite}) that
are found in the algorithm of~\cite{TidenArnborg87} could be missed by 
the current algorithm due to not creating the same number of new variables.
This is shown not to be the case in the following lemma.
\begin{lemma}\label{catch_inf_systems}
Cycles in the sum propagation graph of~\cite{TidenArnborg87} for a system
$S$ in standard form imply cycles in the $\mathcal{LP}$ graph for $S$.  
\end{lemma}
\begin{proof}
 Clearly if the cycle exists in the initial system or is created by one or more
applications of the cancellation rules (a)--(c) then the same cycle will
be created in the $\mathcal{LP}$ graph.\\
Thus, assume the the cycle is created by creating new equations by rule (d).
That is, by rule (d) the following equations are created
\[X =^{?} X_1 + X_2, ~X=^{?}X_3 \times X_4\] 
where $X_1, \ldots, X_4$ are newly created variables. But, then we need to equate
$X_4$ and $X_2$. Equating variables can only happen through rules (b) and (c)
and would require two pre-existing equations of the form 
$X=^{?}L_1 + L_2, ~X=^{?}L_3 \times L_2$ but this is already a cycle
in the $\mathcal{LP}$ graph.
\end{proof}

\subsection{Complexity}\label{complexity_section}
We establish the polynomial time bound in this section.

\begin{lemma}\label{equiv_classes}
The number of $\sim_r^{}$ equivalence classes never increases.
\end{lemma}
\begin{proof}
Rule $(viii)$ is the only rule that creates new variables but 
these variables are contained in pre-existing equivalence classes.
\end{proof}

\begin{lemma}\label{num_sinks}
The number of sinks in any equivalence class after processing is at most one. 
Besides, every non-sink node in the 
class has exactly one outgoing edge.
\end{lemma}
\begin{proof}
If there is no sink in the class, then 
this implies a cycle and thus a non-unifiable system. 
Therefore, let us assume there is no cycle and thus at least one sink. 
In addition, there must be at least one source node.  
It can be seen that rules $(ii)$, $(iii)$, $(iv)$, $(v)$ and
$(vi)$ ensure that all the nodes in the class have at most a single
outgoing edge.
\end{proof}

We now prove several small results about new variables and new lateral
edges that will be useful in the complexity result.

\begin{lemma}\label{var_add}
The maximum number of new variables added to the system $S$ 
is equal to twice the number of equivalence classes.
\end{lemma}
\begin{proof}
Rule $(viii)$ is the only rule that can add variables and this rule can
only add two variables for each sink. By Lemma~\ref{num_sinks} there
is a single sink for each class. By Lemma ~\ref{equiv_classes} the
number of equivalence classes never increases. Rule $(viii)$ adds two new
variables to a lower class for each sink in the upper class.  A class
is only processed one time, thus the number of variables that can be added
is double the number of equivalence classes.
\end{proof}

\begin{lemma}\label{eqdges_to_a_class}
Let $\left[ X \right]_r$ be a $\sim_r^{}$-equivalence class.
Assume there exist $K$ $\sim_r^{}$-equivalence classes one level above
$\left[ X \right]_r$ by the $\gtrdot_r$ ordering. Denote 
the $K$ classes as $C_1, C_2, \ldots, C_K$ and assume that each class
$C_i$ contains $n_i$ variables, such that $N_{K} = \sum_{i=1}^{K} n_i$. 
Then the total number of lateral edges added to 
$\left[ X \right]_r$ by the $K$ higher classes 
is $\leq 2N_{K}^{}$. 
\end{lemma}
\begin{proof}
Processing each $C_i$ will produce $n_i - 1$ edges in $C_i$ each
connecting one node to the sink of that class. If each of these edges
is propagated down by rule $(vii)$ 
or rule $(viii)$ each class could add a
total of $2(n_i - 1)$ edges to $\left[ X \right]_r$. Doing the sum we
get that the $K$ higher classes could add no more then $2N_{K} -2K$
edges to $\left[ X \right]_r$.
\end{proof}

\begin{lemma}\label{edges_add}
The maximum number of lateral edges added to any
$\sim_r^{}$-equivalence class of a system $S$ in standard form is $\mathcal{O}(V_{0}+M)$, 
where $V_{0}$ is the initial number of variables in
$S$ and $M$ is the number of equivalence classes.
\end{lemma}
\begin{proof}  
Follows from Lemma~\ref{eqdges_to_a_class}. 
\end{proof}

These last two lemmas give a bound on the number of edges added to a class
from an outside class, using rules $(vii)$ and $(viii)$. 
We now need to consider the edges added to a class during the processing
of the class itself. 

\begin{lemma}
Let $\left[ X \right]_r$ be a $\sim_r^{}$-equivalence class. 
The number of lateral edges in $\left[ X \right]_r$ does not increase during 
application of step (2) or step (4).
\end{lemma}
\begin{proof}
This follows from the set of inference rules. 
Rules $(vii)$ and $(viii)$ only 
create edges at a lower equivalence class.
The only rule creating a new edge inside the class is rule $(iv)$. 
But rule $(iv)$ also deletes an
edge, thus leaving the number of lateral edges unchanged.
\end{proof}

\begin{lemma}\label{step2_bound}
The number of inference rule applications used during a single application of 
step (2) of Algorithm~\ref{alg_A_1} is bounded  $\mathcal{O}(N + E)$,
where $N$ is the number of variables/nodes and $E$ the number of edges 
in the $\mathcal{LD}$-graph at the start of step (2). 
\ignore{
The number of inference rule applications used during a single application of 
step (2) of Algorithm~\ref{alg_A_1} is polynomially bounded by the number of
edges in the $\mathcal{LD}$-graph at the start of step (2). 
}
\end{lemma}
\begin{proof}
Clearly rules $(i)$ - $(iii)$ are linearly bounded by the number of
edges and $(0)$ by the number of variables.  
\end{proof}

\begin{lemma}\label{equiv_proc}
The number of inference rule applications used to process a single
equivalence class, $C_i$, (step (4) of Algorithm~\ref{alg_A_1}) 
is bounded by $\mathcal{O}(N_i * L_i)$ where $L_i$ is the number of lateral edges and $N_i$ the number nodes 
in the class $C_i$ being processed.
\end{lemma}
\begin{proof}
Rules $(v)$ and $(vi)$ will equate label variables therefore by
Lemma~\ref{label_vars_lemma_2} the number of times they can be applied
is equal to the number of label variables.
Rule $(viii)$ can be applied at most once for each class.
Rule $(vii)$ can be applied once for each variable in the class.
Rule $(x)$ is applied by starting with the
sink node of the path and working back to the start node of the
path, thus if there are $l$ edges in the equivalence class
at the start of the application of rule $(x)$ it will be applied
at most $l$ times. In addition because at the start of the 
application of rule $(x)$ every node, but the sink, has at most
one outgoing lateral edge the number of application of rule
$(x)$ is also bound by $n-1$, where $n$ is the number of nodes in the class.   

Let us now consider rule $(iv)$. 
Let $l_i$ be the number of edges of the equivalence class
$C_i$ to be processed, including edges added from higher classes,
and let $n_i$ be the number of nodes in $C_i$. 
Let us also denote a node for which
rule $(iv)$ is applicable as a $(iv)$-peak.
That is a \textit{(iv)-peak} is a node, $X$, with two edges
leaving $X$ such that the inference rule $(iv)$ is satisfied.
Note that a single node can form more then one $(iv)$-peak.
Now, if we apply rule $(iv)$ to each node $X$ forming a $(iv)$-peak, 
based on the lateral edge partial ordering (by $\succ_{r_*}$), until $X$ is no longer
a peak we have removed $X$ from the set of nodes forming $(iv)$-peaks.
The number of times we can apply rule $(iv)$ to each node is bounded 
by the number of edges leaving that node.   
It can be seen
that each application of rule $(iv)$ removes a $(iv)$-peak but
could add a new $(iv)$-peak. But, the new $(iv)$-peak will be lower
in the $\succ_{r_*}$ path from the initial $(iv)$-peak node to the sink.
As each path must end in a sink the number of these new peaks is naturally
bounded by the length of the path. In addition, because rule $(iv)$ both
removes and adds and edge it can not increase the number of peaks.
Therefore, we can make the following worst case
assumption. Assume that rule $(iv)$ can be applied $l_i$ times
to each node. Then, as rule $(iv)$ will remove one node from the set of peaks
after $l_i$ applications the total number of applications
of rule $(iv)$ in $C_i$ is $\leq n_i * l_i$.    
\end{proof}

We have bounds on the number of classes, the number of new edges and
nodes, and the number of applications of the inference rules. Finally,
we need to bound the size of the $SLPs$. Recall Definition~\ref{slp_def}
for the size of a $SLP$.

\begin{lemma}\label{largest_slp}
The largest, in size, $SLP$ constructed 
by Algorithm~\ref{alg_A_1} on any unification problem~$S$ 
is $\mathcal{O}(|S|^{4})$ where $|S|$ is the initial number of equations. 
\end{lemma}
\begin{proof}
 Assume that we have $M$ topologically sorted, by $\gtrdot_r$,
$\sim_r^{}$-equivalence classes,\\ 
$C_1, C_2, \ldots, C_M$, each
containing $l_1, l_2, \ldots, l_M$ lateral edges, for a total of $L$, and\\ 
$n_1, n_2, \ldots, n_M$, for a total of $N$, variables. 
In addition, let $l_i^{'}$ denote the number of lateral edges
in class $C_i$ at the start of processing and $n_i^{'}$ the
number of variables at the start of processing. $n_i^{'}$ and
$l_i^{'}$ may differ from $n_i$ and $l_i$ because nodes and edges can be added when processing classes above $C_i$.

We need to consider both rule $(x)$ and rule $(iv)$ as these are the 
rules that can add new grammar productions and create larger $SLPs$.
Recall two facts about these two rules, given two $SLPs$  $I = (\Sigma, N_I, P_I)$ and
$J = (\Sigma, N_I, P_J)$.
\begin{enumerate}
 \item For rule $(x)$, creating the new $SLP$ $K$, by Lemma~\ref{slp_concat_bound}
$ | K | = |P_I \cup P_J| + 1 $ and 
$ depth(K) \leq max \left\lbrace  depth(I), depth(J) \right\rbrace  +  1 $. 
\item For rule $(iv)$, creating the new $SLP$ $K$, by Lemma~\ref{slp_suffix_bound}
$ | K | \leq |I| + depth(I) $ and $ depth(K) \leq depth(I) $.
\end{enumerate}
For the analysis we assume that for rule $(x)$  
$depth(K) = max \left\lbrace  depth(I), depth(J) \right\rbrace  +  1$
and for rule $(iv)$ $ | K | = |I| + depth(I) $.
Therefore, rule $(x)$ adds just one new grammar production 
and rule $(iv)$ adds
$depth(I)$ new grammar productions. 
We want to give a bound on the grammar productions created at each level
in the sort of classes and thus the
largest $SLP$ produced will be bounded by the total number of
\emph{unique} productions. 

\emph{Compute the Maximum Depth:}
First note that the $depth$ of any $SLP$ is only increased 
by rule $(x)$ and only by $1$.
Let us first examine the $depth$ of a $SLP$ in a class $C_i$.
Let $D_i$ be the $depth$ of the largest, in $depth$, $SLP$ in $C_i$
at the start of processing. Then, since the application of rule
$(x)$ is bound by $n-1$, where $n$ is the number of nodes in the class, 
the largest, in $depth$, $SLP$ produced in $C_i$ by rule $(x)$ is
\begin{equation}\label{eq_one}
(n_i^{'} -1 ) + D_i 
\end{equation}
where if $C_i$ is a source node in the $\gtrdot_r$ ordering, $D_i = 1$. 
Now assume there are $k$ classes, denoted as $C_{j_1}, \ldots, C_{j_k}$, above $C_i$ in the $\gtrdot_r$ ordering.
Thus, from the $i-1$ classes above $C_i$ in the sort, at the start of processing $C_i$, 
$k$ of them are related to $C_i$ by $\gtrdot_r$. As the classes not related to $C_i$
by $\gtrdot_r$ will not contribute any nodes or edges to $C_i$
we need only consider the $k$ classes.\\ 
\underline{Claim 1:}
\[D_i \leq \sum_{x=1}^{k}(n_{j_x}^{'} -1) +1 \]
\underline{Proof of Claim 1:}\\
First, there must exist at least one ``source'' class in the $k$ classes. By
the bound given in~(\ref{eq_one}), the more source classes we have the less $depth$ we add
as for each source class, $C_s$, $D_s = 1$. Thus, for the worst case analysis let us assume there
is only one source class, say $C_{j_1}$, from the $k$ classes.    
Second, for the worst case analysis when processing any one of the $k$ class 
we want to ensure we are always adding $depth$ to the previous largest, in $depth$, $SLP$.
Thus, assume that the $k$ classes form a chain, each
class adding the maximum number of productions to the largest, in depth, $SLP$ passed down
from the class above and then passing that new $SLP$ to the next class. The process starts
with class $C_{j_1}$ and ends at $C_i$, i.e., like a total ordering.
If we compute the depth of the final deepest $SLP$ in this chain,
using~(\ref{eq_one}) as a bound of the depth at each level, we obtain the following bound on maximum depth.
 \begin{equation}\label{eq_two}
\sum_{x=1}^{k}(n_{j_x}^{'} -1) +1
\end{equation}
\hfill $\Box$ 
\ignore{
Therefore, the depth of the largest $SLP$ at the start of processing $C_i$ is given by~\ref{eq_two}. }

\emph{Compute the Maximum Size:}
Now let us consider rule $(iv)$ on the same class $C_i$.
Lemma~\ref{equiv_proc} bounds the number of applications of rule $(iv)$ for any class based
on $l_i^{'}$ and $n_i^{'}$.  
From Lemma~\ref{eqdges_to_a_class}, we can make the worst case assumption, 
$l_i^{'} = l_i+2 \sum_{x=1}^{k}(n_{j_x} -1)$.
This results in a larger then worst-case bound
for the number of applications of rule $(iv)$ on the
equivalence class $C_i$: 
 \begin{equation}\label{eq_three}
n_i^{'}(l_i+2 \sum_{x=1}^{k}(n_{j_x} -1)) 
\end{equation} 
By Lemma~\ref{slp_suffix_bound}
rule $(iv)$ can add up to $depth(\pi)$ new grammar productions
when applied to a $SLP$ $\pi$. We can make a worst-case assumption that each time
rule $(iv)$ is applied the $SLP$ it is applied to has the maximum
depth. Therefore, combining the bound~(\ref{eq_three}) 
with~(\ref{eq_two}) 
we get the following bound on the number of new 
grammar productions rule $(iv)$ can add during processing of a class
$C_i$:
\begin{equation}\label{eq_four}
 n_i^{'}(l_i+2 \sum_{x=1}^{k}(n_{j_x} -1)) (\sum_{x=1}^{k}(n_{j_x}^{'} -1) +1) 
\end{equation}
We also have from~(\ref{eq_one}) 
that the number of new grammar productions
produced by rule $(x)$ during processing of class $C_i$ is bounded
by
\begin{equation}\label{eq_five}
n_i^{'} -1
\end{equation}
Let us make an additional worst-case assumption, that
each edge in each initial class contains a unique single production $SLP$.
This of course cannot happen as the initial number of unique $SLPs$
before processing for all classes combined is the number of label variables.
Combining this assumption about the unique starting edges with the number 
of applications of $(x)$ 
and the number of new grammar productions created by
$(iv)$ we get the following bound on the number of possible new grammar
productions added by the processing of a class $C_i$:
\begin{equation}\label{eq_six}
 n_i^{'} -1 + n_i^{'}(l_i+2 \sum_{x=1}^{k}(n_{j_x} -1)) (\sum_{x=1}^{k}(n_{j_x}^{'} -1) +1)
\end{equation}
Therefore, to get the total number we add up this value from each class from $1$ to $M$.
Lemma~\ref{var_add} implies that the total number of new variables added to the system
is $\leq 2M$, thus we can assume that $\sum_{x=1}^{k}(n_{j_x} -1) \leq (N+2M)$ and
$\sum_{x=1}^{k}(n_{j_x}^{'} -1) +1 \leq (N + 2M)$. Recall that $N$ is the total number
of initial variables. 
With these assumptions for any class $C_i$, $1 \leq i \leq M$:
\[n_i^{'} -1 + n_i^{'}(l_i+2 \sum_{x=1}^{k}(n_{j_x} -1)) (\sum_{x=1}^{k}(n_{j_x}^{'} -1) +1) \]
\[ \leq (N +2M) + (N+2M)(L+2(N+2M))(N+2M) \] 
Therefore, adding $M$ of these we get:
\[ M \left[  (N +2M) + (N+2M)(L+2(N+2M))(N+2M) \right]  \] 
Since the equations are in standard form there are at most $3$ variables per equation. This implies that $N$, $L$ and $M$ are
$\leq$ $3|S|$, where $|S|$ is the total number of equations. Therefore, 
we get the upper bound $\mathcal{O}(|S|^{4})$.
\end{proof}

\begin{definition}
Let $P_{slp}$ denote the largest polynomial which bounds the run-time for any of the required $SLP$ operations. 
This polynomial is in terms of the largest $SLP$, which by Theorem~\ref{largest_slp}
is $\mathcal{O}(|S|^{4})$.
\end{definition}

From~\cite{Lifshits07} we could assume that $P_{slp} = \mathcal{O}(n^3)$, where
$n$ is the size of the largest $SLP$. 

\begin{theorem}\label{complexity_result}
The worst-case running time of Algorithm~\ref{alg_A_1}
is $\mathcal{O}(|S|^{4} * P_{slp}(|S|^{4}))$, where $|S|$ is the initial number
of equations in standard.
\end{theorem}
\begin{proof}
First let $V$ denote the initial number of label variables and let 
$M$ denote the initial number of equivalence classes. 
Lemma~\ref{label_vars_lemma_2} shows that $V$ does not increase and 
Lemma~\ref{equiv_classes} shows that $M$ does not increase. 

First we consider a general overview of the run time of the algorithm.
\begin{enumerate}
\item Since $V$ does not increase and each time the algorithm returns to 
step (1) it equates two label variables, thus decreasing $V$,  step (1) 
is applied at most $V$ times.
\item Since the algorithm process an equivalence class once and $M$ does
not increase the algorithm applies steps (2) though (5) a maximum of $M$ times.  
\end{enumerate}
Now let us assume that each step (1) through (5), has an associated polynomial,
$P_i$, $i \in \{ 1,2,3,4,5 \}$, which bounds the maximum run-time of that step
in terms of $|S|$ the initial number of equations in standard form.
Based on the general observations above we get the following polynomial, $\mathcal{P}$, 
which bounds the running time of  Algorithm~\ref{alg_A_1}
\begin{equation}\label{bound_gen}
 \mathcal{P} = V(P_{1} + M(P_{2} +P_{3} + P_{4})) 
\end{equation}

It remains to be shown that each $P_i$ is indeed a polynomial in $|S|$ which bounds
the run-time of step $i$. Before examining each step in more detail we present a few useful 
facts. 
\begin{itemize}
\item First, rule $(0)$ is the only rule that affects the initial set of equations and this rule only
equates two variables. Therefore, we can bound the run-time by $|S|$ without concern that $|S|$ will
increase.
\item Let $V_{0}$ denote the set of all variables in the initial standard form set of equations. Then, 
by the structure of the equations we can see that there are at most $3$ variables per equation and
$V_{0} \leq 3 * |S|$. This also implies that $V \leq 3 * |S|$ and $M \leq 3 * |S|$. 
\item It also easily follows that $L$ the number of lateral edges is bounded by $|S|$, in fact
$L \leq |S|$. 

\end{itemize}
We now consider each step in Algorithm~\ref{alg_A_1}.
\begin{enumerate}
\item Step (1), by standard graph construction methods, is bounded by $\mathcal{O}(V_{0} * |S|)$, which results 
in $P_1 \leq \mathcal{C}_1 * |S|^{2}$, for some constant~$\mathcal{C}_1$.
\item By Lemma~\ref{step2_bound} the number of  inference rules applied at step (2) is
$\mathcal{O}(N +E)$, where $N$ is the number of nodes and $E$ the number of edges in the
$\mathcal{LD}$ graph. Let $V_{0}$ denote the initial number of variables in $S$, then
 $N \leq V_{0} + 2*M$, since by Lemma~\ref{var_add} the maximum number of variables
 that can be added is twice the number of equivalence classes. Next, by
 Lemma~\ref{edges_add}
The maximum number of lateral edges added to any
equivalence class is $\mathcal{O}(V_{0}+M)$. 
We can thus conservatively say that the maximal number of lateral
edges is $\mathcal{O}(M(V_{0}+M))$.
Since the number of downward edges does not change we get the bound of
$C_2(V_{0} + 2*M) + M(V_{0}+M))$. Rewriting in terms of $|S|$,
$P_2 \leq \mathcal{C}_2 *(|S|^{2} + |S|^{2} + |S|^{2})* P_{slp}(|S|^{4})$, for some constant~$\mathcal{C}_2$.
\item Using standard graph cycle checking we get that
$P_3 \leq \mathcal{C}_3 * |S|$, for some constant~$\mathcal{C}_3$.
\item By Lemma~\ref{equiv_proc},
the number of rules applied for class $i$ is
$\mathcal{O}(N_i * L_i)$ where $L_i$ is the number of lateral edges and $N_i$ the number nodes 
in the class $i$ being processed. Thus for each class we get a run-time bound of $\mathcal{O}((N_i * L_i)* P_{slp}(|S|^{4}))$.
Rewriting in terms of $|S|$ we get $P_4 \leq \mathcal{C}_4 * |S|^{2} P_{slp}(|S|^{4})$, for some constant $\mathcal{C}_4$.
\end{enumerate}
Now plugging all these into Equation~(\ref{bound_gen}),  letting $\mathcal{C} = Max(\mathcal{C}_1, ~\mathcal{C}_2, ~\mathcal{C}_3, ~\mathcal{C}_4)$ 
and replacing $V$ and $M$ in terms of $S$ we get
\ignore{
\begin{equation}
\mathcal{P} \leq V((\mathcal{C}_1 * |S|^{2}) + M((\mathcal{C}_2 *(|S|^{2} + |S|^{2} + |S|^{2})* P_{slp}(|S|^{4}) ) 
+(\mathcal{C}_3 * |S| ) + (\mathcal{C}_4 * |S|^{2} * P_{slp}(|S|^{4}) ))) 
\end{equation}
Letting $\mathcal{C} = Max(\mathcal{C}_1, ~\mathcal{C}_2, ~\mathcal{C}_3, ~\mathcal{C}_4)$ 
and replacing $V$ and $M$ in terms of $S$ we get
}
\begin{equation}
\mathcal{P} \leq \mathcal{C} * (|S|^{3} + |S|^{4}* P_{slp}(|S|^{4})  +|S|^{3}  + |S|^{4} * P_{slp}(|S|^{4} ) )
\end{equation}
or 
\begin{equation}
\mathcal{O}( |S|^{4} * P_{slp}(|S|^{4} ) )
\end{equation}

\end{proof}

\section{On Asymmetric Unification and One-Sided Distributivity}\label{asym_unif_sec}
Our work on a polynomial bounded procedure was partially motivated 
by its potential application to cryptographic protocol analysis. Since our initial results~\cite{marshall12c}, 
a new unification paradigm has been developed in~\cite{Erbatur13a} and is based on newly identified requirements arising
from the symbolic analysis of cryptographic protocols.
The analysis involves the unification-based exploration of a space in which the states obey equational theories that can be expressed as a decomposition $R \uplus E$, where $R$ is a set of rewrite rules that is confluent, terminating and coherent modulo $E$. 
In order to apply state space reduction techniques, it is usually necessary for at least part of this state to be in normal form, and  to remain in normal form even after unification is performed.  This requirement can be expressed as an {\em asymmetric} unification
problem $\{ s_1 =^{\downarrow} t_1, ~\ldots, ~s_n =^{\downarrow} t_n \} $ where the $=^{\downarrow}$ denotes a unification problem with the  restriction that any unifier leaves the
right-hand side of each equation irreducible.

Let us review a few definitions needed for asymmetric unification problems. 
A \emph{rewrite rule} is an ordered pair $l \rightarrow r$ such that $l, r \in T(\Sigma, \mathcal{X})$ and $l \not\in \mathcal{X}$.
The rewrite relation on $T(\Sigma, \mathcal{X})$, written $t  \rightarrow_R s$, holds between $t$ and $s$ iff there 
exists a non-variable $p \in Pos_{\Sigma}(t)$, $l \rightarrow r \in R$ and a substitution $\delta$, such that $t |_p = l \delta$
and $s =t[r \delta]_p$.
The relation $\rightarrow_{R/E}$ on $T (\Sigma, \mathcal{X})$ is $=_{E} \circ \rightarrow_{R} \circ =_{E}$. 
The relation $\rightarrow_{R,E}$ on $T (\Sigma, \mathcal{X})$ is defined as: 
$t \rightarrow_{R,E} t'$ if there exists a position $p \in Pos_{\Sigma}(t)$, a rule $l \rightarrow r \in R$ and
a substitution $\delta$ such that $t|_{p} =_{E} l \delta$ and $t' = t[r \delta]_{p}$.
The transitive (resp. transitive and reflexive) closure of $\rightarrow_{R,E}$ is denoted by $\rightarrow^{+}_{R,E}$ 
(resp. $\rightarrow^{*}_{R,E}$). 
A term $t$ is $\rightarrow_{R, E}$ irreducible if there is no term $t'$ such that $t \rightarrow_{R,E} t'$. $t$ is then said to be a $R,E$-normal form (or just normal form).
 If $\rightarrow_{R,E}$ is confluent and terminating we denote the irreducible version of a term 
 $t$ by  $t \downarrow_{R,E}$.
\begin{definition}
We call $(\Sigma, ~E, ~R)$ a \emph{decomposition} of an equational theory $\Delta$ over
a signature $\Sigma$ if $\Delta = R \uplus E$ and $R$ and $E$ satisfy the following conditions:
\begin{enumerate}
\item $E$ is variable preserving, i.e., for each $s = t$ in $E$ we have $Var(s) = Var(t)$.
\item $E$ has a finitary and complete unification algorithm.
\item For each $l \rightarrow r \in R$ we have $Var(r) \subseteq Var(l)$.
\item $R$ is confluent and terminating modulo $E$, i.e., the relation
$\rightarrow_{R/E}$ is confluent and terminating.
\item $R$ is coherent modulo $E$, i.e., $\forall t_1, t_2, t_3$ if $t_1 \rightarrow_{R,E} t_2$ and 
$t_1 =_{E} t_3$ then $\exists ~t_4, t_5$ such that $t_2 \rightarrow^{*}_{R, E} t_4$, $t_3 \rightarrow^{+}_{R, E} t_5$, and
$t_4 =_{E} t_5$.  
\end{enumerate}
\end{definition}  

\begin{definition}
\textbf{(Asymmetric Unification)}. Given a decomposition $(\Sigma, E, R)$
of an equational theory, a substitution $\delta$
is an \emph{asymmetric $R,E$-unifier} of a set $\mathcal{S}$ of asymmetric equations
$\{ s_1 =^{\downarrow} t_1, ~\ldots, ~s_n =^{\downarrow} t_n \} $ iff for each
asymmetric equations $s_i =^{\downarrow} t_i$, $\delta$ is an $(E \cup R)$-unifier
of the equation $s_i =^{?} t_i$ and $(t_i \downarrow_{R,E})\delta$ is in $R,E$-normal
form. A set of substitutions $\Omega$ is a \emph{complete set of asymmetric} 
$R, E$-unifiers of $\mathcal{S}$ (denoted $CSAU(\mathcal{S})$) iff: (i) every member of $\Omega$ is an asymmetric 
$R,E$-unifier of $\mathcal{S}$, and (ii) for every asymmetric $R,E$-unifier $\theta$ of
$\mathcal{S}$ there exists a $\delta \in \Omega$ such that $\delta \leq_{E}^{Var(\mathcal{S})} \theta$.
\end{definition}

\begin{example}
Let $R = \{ X \oplus 0 \rightarrow X, ~X \oplus X \rightarrow 0 , ~X \oplus X \oplus Y \rightarrow Y \}$
and $E$ be the associativity and commutativity~(AC) axioms for $\oplus$.
Consider the equation $Y \oplus X =^{\downarrow} X \oplus a$.
The substitution $\delta_1 = \{ Y \mapsto a \}$ 
is an asymmetric solution since the right hand side will
remain irreducible after applying $\delta_1$. 
But, $\delta_2 = \{ Y \mapsto a, ~X \mapsto 0 \}$ 
is not an asymmetric unifier, 
although it is a unifier, 
since $0 \oplus a \rightarrow_{R, E} a$. 
\end{example}

We consider the one-sided distributivity theory in this new asymmetric setting. First, we need to present the
axioms as a theory decomposition. In this case the theory decomposition is simple.
Let  $\Delta= R \cup E$, where $R= \{ X \times (Y + Z) \rightarrow X \times Y + X \times Z \}$ and $E= \emptyset$.

One way of approaching the asymmetric unification problem is to start with the
symmetric unifiers and then try modifying them, if need be, into asymmetric
unifiers. Thus we could have first obtained the symmetric unifier using the
original Tid\'{e}n-Arnborg algorithm. This method looks feasible as far as
decidability is concerned, but instead we develop an algorithm where
failures can be detected much earlier.

In what follows we are going to assume that variables are 
always mapped to $R, \emptyset$-normal forms.
We can do this by assuming, without loss of generality, that all substitutions are $R, \emptyset$-normalized.

Based on $\Delta$, the following inference rules represent an asymmetric algorithm and are 
a simple modification of the original Tid\'{e}n-Arnborg algorithm to the new asymmetric domain. 
The soundness of the rules follow directly from the rules presented in Section~3. In addition,
since the asymmetric restriction does not affect the system being subterm collapse free, the error conditions of the original 
algorithm, and the graphs used to detect them, remain unchanged.  The only additional 
error conditions, rules (e) and (f), follow due to the rewrite rule, $R= \{X \times (Y + Z) \rightarrow X \times Y + X \times Z \}$ , 
which requires that we apply a reduction to the $\times$ rooted term.
Likewise, rules (e') and (f') would imply failure because a reduction
could be applied to a term with an irreducible restriction. 
We denote the algorithm, consisting of the following 
inference rules along with the error checking, as 
\emph{Algorithm 3}.

\begin{figure}[ht]
  \centering
\scalebox{.95}{
\fbox{
\begin{small}
\begin{tabular}{lcc}
(a) & $\qquad$ & $\vcenter{
\infer[\qquad \mathrm{if} ~ U ~ \mathrm{occurs ~ in} ~ \eq ]{\{U =_{}^{\downarrow} V\} \cup \, \{ U \mapsto V\}(\eq) }
      { \eq ~ \uplus ~\{U =_{}^{\downarrow} V\} }
}
$\\[+20pt]
(b) & & $\vcenter{
\infer[\qquad \mathrm{where} ~\circ ~\mathrm{is ~either} ~+ ~\mathrm{or} ~ \times ]{\eq ~ \cup ~ 
\{ U =_{}^{\downarrow} V \circ W, \; X =_{}^{\downarrow} V, \; Y =_{}^{\downarrow} W \}}
{\eq ~ \uplus ~ \{ U =_{}^{\downarrow} V \circ W, \; U =_{}^{\downarrow} X \circ Y \}}
}
$\\[+20pt]
(c) & & $\vcenter{
\infer{\eq ~ \cup ~ \{ U =_{}^{\downarrow} V \circ W, \; X =_{}^{\downarrow} V, \; Y =_{}^{\downarrow} W \}}
{\eq ~ \uplus ~ \{ U =_{}^{\downarrow} V \circ W, \; X \circ Y =_{}^{\downarrow} U   \}}
}
$\\[+20pt]
(d) & & $\vcenter{
\infer{\eq ~ \cup ~ \{ V \circ W =_{}^{\downarrow} U, \; X =_{}^{\downarrow} V, \; Y =_{}^{\downarrow} W \}}
{\eq ~ \uplus ~ \{ V \circ W =_{}^{\downarrow} U, \; X \circ Y =_{}^{\downarrow} U   \}}
}
$\\[+20pt]
(e) & & $\vcenter{
\infer{FAIL}
{\eq ~ \uplus ~ \{ U =_{}^{\downarrow} V \times W, \; U =_{}^{\downarrow} X + Y \}}
}
$\\[+20pt]
(e') & & $\vcenter{
\infer{FAIL}
{\eq ~ \uplus ~ \{ U =_{}^{\downarrow} V \times W, \; 
W =_{}^{\downarrow} X + Y \}}
}
$\\[+20pt]
(f) & & $\vcenter{
\infer{FAIL}
{\eq ~ \uplus ~ \{ U =_{}^{\downarrow} V \times W, \; 
X + Y =_{}^{\downarrow}  U \}}
}
$\\[+20pt]
(f') & & $\vcenter{
\infer{FAIL}
{\eq ~ \uplus ~ \{ U =_{}^{\downarrow} V \times W, \; 
X + Y =_{}^{\downarrow}  W \}}
}
$\\[+20pt]
(g) & & $\vcenter{
\infer{\eq ~ \cup ~ \{ V \times W =_{}^{\downarrow} U, \; W_1^{} + W_2^{} =_{}^{\downarrow} W , 
\; V \times W_1^{} =_{}^{\downarrow} X, \; V \times W_2^{} =_{}^{\downarrow}  Y \}}
{\eq ~ \uplus ~ \{ V \times W =_{}^{\downarrow} U, \; U =_{}^{\downarrow} X + Y \}}
}
$\\[+20pt]
(h) & & $\vcenter{
\infer{\eq ~ \cup ~ \{ V \times W =_{}^{\downarrow} U, \; W_1^{} + W_2^{} =_{}^{\downarrow} W  , 
\; V \times W_1^{} =_{}^{\downarrow} X, \; V \times W_2^{} =_{}^{\downarrow}  Y \}}
{\eq ~ \uplus ~ \{ V \times W =_{}^{\downarrow} U, \; X + Y =_{}^{\downarrow} U  \}}
}
$
\end{tabular}
\end{small}
} }
 \caption{Asymmetric Inference Rules for \emph{Algorithm 3}.}
  \label{Asym_Inf_Rules}
 \end{figure} 

In addition to error checking remaining the same, the soundness of the above procedure can be shown by 
showing each rule is sound and this follows since each rule is just an asymmetric instantiation of the 
sound symmetric rules presented in Section~\ref{TA_Section}. In addition, we can assume 
termination since the original algorithm is terminating. 

\ignore{
\begin{theorem}(Tid\'{e}n and Arnborg~\cite{TidenArnborg87})\label{TA_solutions}
\begin{itemize}
\item $U =^{?} X + Y ~\land ~U=^{?} V + W$ have the same set of solutions as\\
$U =^{?} X + Y ~\land ~X=^{?} V ~\land ~Y=^{?} W$.
\item $U =^{?} X \times Y ~\land ~U=^{?} V \times W$ have the same set of solutions as\\
$U =^{?} X \times Y ~\land ~X=^{?} V ~\land ~Y=^{?} W$. 
\item $U =^{?} X + Y ~\land ~ U=^{?} V \times W$ have the same set of solutions as\\
$\exists W_1, W_2:  U=^{?} V \times W ~\land ~ W=^{?}W_1 + W_2 ~\land~ 
X=^{?} V \times W_1 ~\land ~ Y=^{?} V \times W_2$.
\end{itemize}
\end{theorem}
}

In the following let $\circ$ denote either $+$ or $\times$.
\begin{lemma}\label{A2_rules1}
The set of equations \\
$\{ U =^{\downarrow} V \circ W, ~U=^{\downarrow} X \circ Y \}$\\
and the set of equations\\
$\{ U =^{\downarrow} V \circ W, ~ X \circ Y =^{\downarrow} U \}$\\
have the same asymmetric solutions as the set \\
$\{ U =^{\downarrow} V \circ W, ~X =^{\downarrow} V ,~ Y=^{\downarrow} W \}$.
\end{lemma}
\begin{proof}
The fact that the equations have the same unifiers is a result of Theorem~\ref{TA_sound}.
Next note that the asymmetric restrictions are maintained since 
the instances of $V \circ W$, 
$V$ and $W$ must be irreducible. 
\end{proof}

\begin{lemma}\label{A2_rules2}
The set of equations\\
$\{ V \circ W=^{\downarrow} U, ~ X \circ Y =^{\downarrow} U \}$\\
has the same asymmetric solutions as the set \\
$\{ V \circ W=^{\downarrow} U, ~X =^{\downarrow} V, ~ Y=^{\downarrow} W \}$.
\end{lemma}
\begin{proof}
The fact that the equations have the same unifiers is a result of Theorem~\ref{TA_sound}.
We can, without loss of generality, assume that all substitutions are $R, \emptyset$ 
normalized. This implies that variables are always mapped to 
$R, \emptyset$-normal forms and we
can apply an irreducibility restriction to them without restricting the solution space. This implies
the correctness of the last two equations $X =^{\downarrow} V, ~ Y=^{\downarrow} W$. 
\end{proof}

\begin{lemma}\label{A2_rules3}
The following sets of equations have no asymmetric $R, \emptyset$ solutions:\\
$\{ U =^{\downarrow} V \times W, ~U=^{\downarrow} X + Y \}$,\\
$\{ U =^{\downarrow} V \times W, ~X + Y=^{\downarrow} U \}$.
\end{lemma}
\begin{proof}
This is due to the orientation of $R$ which requires a reduction in the $\times$-rooted
equation in order to move a $+$ to the top.
\end{proof}

\begin{lemma}\label{A2_rules4}
The set of equations\\
$\{ V \times W=^{\downarrow} U, ~ U=^{\downarrow} X + Y \}$\\
and the set of equations\\
$\{V \times W=^{\downarrow} U, ~ X + Y=^{\downarrow} U  \}$\\
have the same asymmetric solutions as the set\\
$\{V \times W=^{\downarrow} U, ~W_1 + W_2 =^{\downarrow} W,
 ~ V \times W_1 =^{\downarrow} X, ~V \times W_2 =^{\downarrow} Y \}$, where
$W_1$ and $W_1$ are fresh variables.
\end{lemma}
\begin{proof}
The fact that the equations have the same unifiers is a result of Theorem~\ref{TA_sound}.
In addition, since we can assume that variables are 
mapped to $R, \emptyset$-normal forms and 
all substitutions are normalized we a free to place an irreducibility restriction on a variable
without reducing the set of solutions.
\end{proof}

Recall that in addition to the two failure rules we maintain the two graphs used to detect
failure in the original  symmetric algorithm. 

\begin{lemma}\label{A2_errors_correct}
The error conditions of Algorithm 3 are correct.
\end{lemma}
\begin{proof}
This follows from Lemma~\ref{A2_rules3} and the fact that
adding the asymmetric restriction does not change the fact that the theory is still 
simple, i.e., cycles imply failure. Thus, the use of the graph based method to detect cycles is still valid.
In addition, it has been shown in~\cite{TidenArnborg87} that systems that cause non-termination are not unifiable and are detectable via the sum propagation graph method. 
Since adding irreducibility constraints does not increase the types of systems which cause non-termination and the systems are still detectable via the graph method, the use of a propagation graph to detect all non-terminating systems is still correct.
\end{proof}

\begin{theorem}\label{A2_correct}
Asymmetric $R, \emptyset$ unification is decidable.
\end{theorem}
\begin{proof}
Lemma~\ref{A2_errors_correct} shows that
if a system is not asymmetrically unifiable it will be detected by one or more of the failure rules or graphs. 
In addition, Lemma~\ref{A2_errors_correct} shows that non-terminating systems
will also be detected implying Algorithm 3 terminates. 

Lemma~\ref{A2_rules1}, Lemma~\ref{A2_rules2} and Lemma~\ref{A2_rules4} show that
Algorithm 3 transforms a system into a solved-form maintaining the set of solutions.
This implies that the substitution constructed from the final solved form is an 
asymmetric solution to the initial problem. 
\end{proof}

\begin{theorem}\label{A2_CSAU}
Algorithm 3 produces a complete set of asymmetric unifiers. 
\end{theorem}
\begin{proof}
Consider an asymmetric problem, $S$, and its solved form, $S'$ produced by Algorithm~3.
Let $\delta_{S'}$ denote the substitution obtained from $S'$ in the standard way. 
Recall that a substitution obtained from a dag solved form is idempotent, i.e., $\delta_{S'} = \delta_{S'} \delta_{S'}$.
Let $\theta$ be an asymmetric solution to $S$ and let $X \in Var(S)$.

\begin{enumerate}
\item If $X \not \in Var(S')$,  $X \delta_{S'} = X$ and $X \delta_{S'} \theta = X \theta$.
\item If $X \in Var(S')$, then there are two cases.
\begin{enumerate}
\item  $X \delta_{S'} \mapsto X $, in which case  $X \delta_{S'} \theta = X \theta$.
\item $X \delta_{S'} \mapsto t_i$, for some term $t_i$. This implies there is an equation in $S'$ of the
form $X =^{\downarrow} t_i$. Recall Lemma~\ref{A2_rules1}, Lemma~\ref{A2_rules2} and Lemma~\ref{A2_rules4} show that
Algorithm 3 transforms a system into a solved-form maintaining the set of solutions.
Thus, $X \theta =_{\Delta} t_i \theta$. 
This implies that $X \delta_{S'} \theta =_{\Delta}   t_i \theta =_{\Delta}  X \theta$. \qedhere
\end{enumerate}
\end{enumerate}

\end{proof}
\noindent Therefore, we can obtain an asymmetric unification algorithm by modifying the original symmetric algorithm. This new algorithm has the following 
beneficial characteristics:
\begin{itemize}
 \item Much as the original  algorithm of~\cite{TidenArnborg87}, this new algorithm is conceptually easy to grasp,
 and easy to implement.
 \item Again, like the  Tid\'{e}n and Arnborg algorithm, the new asymmetric
 algorithm should perform well computationally on most problem instances,
 since it is unlikely a problem will have the structure needed to force
 the exponential behavior.    
 \end{itemize} 
\subsubsection*{Complexity}
\begin{definition} \label{sigma_2} 
For $n \ge 0$, let $\sigma'(n)$ be the set of equations
\begin{eqnarray*}
X_{1^{i+1}}+X_{1^{i}2} &  =^{\downarrow}  & X_{1^{i}}, \\
Y_{2^{i}1}+Y_{2^{i+1}} &  =^{\downarrow}  & Y_{2^{i}}, \\
T \times X_{1^{i}2} &  =^{\downarrow}  & Y_{2^{i}1},\\
T \times Y &  =^{\downarrow}  & X, \\
X_{1^{i+2}}+X_{1^{i+1}2} &  =^{\downarrow}  & X_{1^{i+1}}  
\end{eqnarray*}
for all $0 \le i \le n$, where $X_{l^{i}}$ denotes $i$ concatenations of $l \in \{1,2\}$, i.e.,
$X_{1^{3}2} = X_{1112}$. 
\end{definition}
A simple modification to the $\sigma(n)$ definition 
(see Definition~\ref{sigma_2}) again results in a family of equations, this time asymmetric, which cause exponentially many applications of the
inference rules. The new definition, $\sigma'(n)$, simply places the irreducibility restriction on the variables
which are already irreducible with respect to the rewrite rule. 
Since we can assume, without loss of generality, that substitutions are fully reduced via $R, \emptyset$-rewriting, the
irreducibility restrictions will not be violated. Therefore, the action of the new algorithm does not
change, in terms of complexity, from the action of the original algorithm on $\sigma(n)$.

\subsubsection*{An Open Problem: Polynomial-Time Decision Algorithms for Asymmetric Unification modulo One-Sided Distributivity}
The polynomial algorithm developed in Section~\ref{general_section} relies 
on the use of $SLPs$ to ensure the polynomial bound. The $SLP$ compression
method can be used because the critical information, path labels, are maintained by the compression method. In addition, there are polynomial bounded algorithms for answering the required questions regarding $SLP$
compressed strings. However, when asymmetric unification is considered 
we are forced to also keep track of the irreducibility restriction.
This information would unfortunately be lost in the current 
compression method. The current compression scheme used in Algorithm~\ref{alg_A_1} would need to be modified, to maintain the irreducibility constraints, before the algorithm could be applied to the asymmetric case.

Therefore, a polynomial time asymmetric algorithm based on compression is still an open problem. There are a couple of possible approaches:
\begin{enumerate}
\item Develop a method of encoding the irreducibility restriction into
the same $SLPs$. This seems like it may be possible, but it also requires ensuring the $SLP$ algorithms used in Section~\ref{slp_ops} can be applied,
in polynomial time, to these new encodings.
\item Use a different compression method. This may also be possible, for example perhaps using the methods developed in~\cite{gascon11}. 
Again, we would need to ensure all the operations used in Section~\ref{slp_ops} could be done on the new compression method in polynomial time.  
\end{enumerate}

\section{Conclusions}
Three problems are solved in this paper:
\begin{enumerate}
\item We have developed a new polynomial 
time algorithm which solves the \emph{decision} problem for
a non-trivial subcase, based on a typed theory, of 
unification modulo one-sided distributivity.
This subcase happens to be sufficient to express the negative 
complexity result in~\cite{narendran10}. The
new algorithm is conceptually easy to understand and more efficient than the new algorithm solving the general problem.
\item We developed the first polynomial time algorithm which
solves the \emph{decision} problem for unification modulo 
one-sided distributivity. 
\item We developed the first algorithm that solves the asymmetric
unification problem for unification modulo 
one-sided distributivity. That is, the algorithm produces the most general asymmetric unifier. 
Although this new algorithm is not polynomial, it is conceptually 
easy to grasp and easily implemented. In addition, it should perform 
well computationally on most problem instances.   
\end{enumerate}
Although the focus of this paper is on decision procedures and complexity we can note that all the algorithms
presented in the paper compute unifiers. In the asymmetric case a complete set of unifiers can be obtained from 
the computed solved forms. In the compressed case, a resulting solved form is actually a 
compressed representation of a unifier. 

\section{Acknowledgments}
We thank the reviewers for their helpful comments, especially their 
suggestions for simplifying Algorithm~\ref{alg_A_1}.
  

\bibliographystyle{plain} 
\bibliography{TA3_4}

\section*{Appendix A: Computing a Suffix}\label{comp_compress_suff_sec}
For completeness we present in a very simple recursive algorithm for computing
the compressed suffix in polynomial time. One could also use the methods
developed in~\cite{gascon09, Levy08, Lifshits07, schmidtschauLIPIcs2012} to efficiently compute
the suffix and prefix. 
The algorithm only requires the size of the string produced by the prefix 
and the actual $SLP$ containing the prefix and suffix. 
We assume also that the size of the string produced for each $SLP$ is contained in the data structure.  
Let $\pi_1$ denote the large $SLP$ containing the suffix and let $\pi_2$ denote the
prefix. Let $diff = || \pi_1 || - || \pi_2 ||$, i.e., $diff$ is the size of the string produced by the
suffix. The algorithm returns the suffix $SLP$, denoted as $\pi_3$. 

\begin{algorithm}
\caption{BuildSuffix}
\label{alg2}
\begin{algorithmic}
\STATE (Input:~$\pi_1$, $diff$)
\STATE  Create a $SLP$ pointer: temp = $\pi_1$
\WHILE{$||RightChild(temp)||$ $\geq$ $diff$}
\item $temp=RightChild(temp)$
\ENDWHILE 
\IF{$||temp|| == d_f$}
\item return $temp$
\ELSE
\item Create new non-terminal $\pi_3$.
\item $RightChild( \pi_3 )=RightChild(temp)$
\item $LeftChild( \pi_3 )=BuildSuffix(LeftChild(temp), diff - || RightChild(\pi_3) ||)$
\item return $ \pi_3 $
\ENDIF
\end{algorithmic}
\end{algorithm}
\ignore{
\begin{lemma}
Given a straight line program $\pi_1$ and a prefix $\pi_2$, 
Algorithm~\ref{alg2} will return a straight line program $\pi_3$
that is the result of quotient operation ${\pi_2}_{}^{-1}  \pi_1$.
\end{lemma}
\begin{proof}
The SLP $\pi_1$ generates exactly one string. It can be seen that
Algorithm~\ref{alg2} removes from the SLP $\pi_1$ the section
that generates the prefix string of length $| \pi_2 |$. Therefore
the remaining SLP generates the suffix $\pi_3$.
\end{proof}
}

\begin{theorem} 
Algorithm~\ref{alg2} runs in $\mathcal{O}(depth(\pi_1))$, 
$| \pi_3 | \leq |\pi_1| + depth(\pi_1) $ and $depth(\pi_3) \leq depth(\pi_1)$. 
\end{theorem}
\begin{proof}
 Consider the recursive call in Algorithm~\ref{alg2}. 
The algorithm only uses a single linear recursive call and the recursion is always called on a non-terminal one level lower in $\pi_1$. 
Therefore, the algorithm is bounded by $depth(\pi_1)$. In addition, a new rule/non-terminal,
is created for each recursive call for a maximum of 
$depth(\pi_1)$ new rules, thus  $| \pi_3 | \leq |\pi_1| + depth(\pi_1)$. 
\end{proof}
 
{\cW Blah}
\vspace{-30 pt}
\end{document}